\documentclass[reqno]{amsart}

\usepackage{breakurl}
\usepackage{hyperref}
\usepackage{amssymb}
\usepackage{amsthm}
\usepackage{mathtools}
\usepackage{enumitem}

\usepackage{xcolor}
\hypersetup{
    colorlinks,
    linkcolor={red!50!black},
    citecolor={blue!50!black},
    urlcolor={blue!80!black}
}

\usepackage{algorithm}
\usepackage{algpseudocode}

\usepackage[
        paperwidth=\textwidth+3cm,
        paperheight=\textheight+4cm,
        text={\textwidth,\textheight} 
    ]{geometry}

\usepackage{multirow}

\algrenewcommand\algorithmicrequire{\textbf{Input:}}
\algrenewcommand\algorithmicensure{\textbf{Output:}}

\newcommand{\Real}{\operatorname{Re}}
\newcommand{\Imag}{\operatorname{Im}}

\newcommand{\ZZ}{\mathbb{Z}}
\newcommand{\QQ}{\mathbb{Q}}
\newcommand{\NN}{\mathbb{N}}
\newcommand{\RR}{\mathbb{R}}
\newcommand{\CC}{\mathbb{C}}

\newcommand{\bigOtilde}{\widetilde O}

\newtheorem{theorem}{Theorem}[section]

\theoremstyle{definition}
\newtheorem{problem}[theorem]{Open problem}

\makeatletter
\@namedef{subjclassname@2020}{%
  \textup{2020} Mathematics Subject Classification}
\makeatother

\begin{document}

\title{Arbitrary-precision computation of the gamma function}

\author{Fredrik Johansson}
\address{Inria Bordeaux, 33400 Talence, France}
\email{fredrik.johansson@gmail.com}



\begin{abstract}
We discuss the best methods available for computing the gamma function $\Gamma(z)$
in arbitrary-precision arithmetic with rigorous error bounds.
We address different cases: rational, algebraic, real or complex arguments;
large or small arguments; low or high precision;
with or without precomputation.
The methods also cover the log-gamma
function $\log \Gamma(z)$, the digamma function $\psi(z)$,
and derivatives $\Gamma^{(n)}(z)$ and $\psi^{(n)}(z)$.
Besides attempting to summarize the existing state of the art, we present
some new formulas, estimates, bounds and algorithmic improvements and discuss implementation results.
\end{abstract}

\maketitle


\section{Introduction}

The gamma function
\begin{equation}
\label{eq:gammadef}
\Gamma(z) = \int_0^{\infty} t^{z-1} e^{-t} dt \quad (\Real(z) > 0), \quad \Gamma(z) = \frac{\Gamma(z+1)}{z},
\end{equation}
is arguably the most important higher transcendental function,
having a tendency to crop up in any setting involving sequences, series, products or integrals
when the solutions go beyond elementary functions.

Calculation of the gamma function is
a classical subject in numerical analysis.
The standard algorithm
uses the asymptotic Stirling series
\begin{equation}
\Gamma(z+1) \,= \,z! \;\sim \;\sqrt{2 \pi} \, e^{-z} z^{z+1/2} \left(1 + \frac{1}{12 z} + \frac{1}{288 z^2} + \ldots \right), \quad |z| \to \infty,
\label{eq:stirlingseries1}
\end{equation}
which is valid in every closed sector of the complex plane excluding the negative real line.
If $z$ is too small or too close to the negative real line to use \eqref{eq:stirlingseries1} directly,
one first applies a shift $z \to z + r$ and possibly the reflection formula
\begin{equation}
\Gamma\!\left(z\right) = \frac{\pi}{\sin\!\left(\pi z\right)} \frac{1}{\Gamma\!\left(1 - z\right)}.
\label{eq:reflection}
\end{equation}

The story does not end here: a close study of the Stirling series
presents a number of
interesting theoretical problems and implementation issues, and
there are also many alternative algorithms.
However, despite an extensive literature on gamma function computation,\footnote{See Davis \cite{Davis1959} for the history of the gamma function until 1959.
Algorithms, software implementations and survey works include \cite{hastings1955approximations,Lanczos1964,Cody1967,Filho1967,luke1969special,Luke1970,Kuki1972,Klbig1972,Ng1975,Brent1978,van1984calculation,Borwein1987,Macleod1989,BorweinZucker1992,Cody1991,Cody1993,Spouge1994,Hare1997,Smith2001,laurie2005,Fousse2007,schmelzer2007computing,cuyt2008handbook,Johansson2014thesis,beebe2017mathematical}; others will be cited below. For more history and bibliography, see \cite{sebah2002introduction,Borwein2018,Olver:2010:NHMF,perez2020notes}.}
there have been surprisingly few
attempts to investigate the best available techniques
from the point of view of arbitrary-precision computation.
It seems worthwhile to pursue this topic
since some applications require repeated
evaluation of $\Gamma(z)$ to
a precision of hundreds or thousands of digits.

\subsection{Quick survey of methods}

It is convenient to sort algorithms for the gamma function
into four broad categories.
\begin{enumerate}
\item \emph{Global} methods approximate $\Gamma(z)$ in the whole complex plane
using its asymptotics
together with a correction.
The Stirling series
and the formulas of Lanczos \cite{Lanczos1964} and Spouge \cite{Spouge1994}
are all of this type. 
\item \emph{Local} methods approximate the gamma function near particular
points or on particular intervals such as $z \in [1, 2]$ using tabulated values,
Taylor series,
Chebyshev interpolants, or other approximations. 
\item \emph{Hypergeometric} methods represent the gamma function
in terms of hypergeometric functions which may be
evaluated using series expansions.
\item \emph{Integration} methods evaluate \eqref{eq:gammadef}
or a suitable contour integral using standard numerical integration
techniques. 
\end{enumerate}

The different categories all have their uses.
Global methods are the most broadly useful,
but local methods can
perform better where they are applicable. The main drawback
of the local methods is their high precomputation cost.
The hypergeometric methods are the most efficient asymptotically at high precision
due to the possibility to use complexity-reduction techniques.

Table~\ref{tab:algoverview} lists some of the major
algorithms along with estimates of
the rate of convergence, expressed as the number of terms needed asymptotically to
achieve $p$-bit accuracy.
This rate cannot be taken as a direct measure of efficiency,
because the notion of ``term'' differs between the algorithms:
\begin{itemize}
\item The ``generic'' terms require at least a full multiplication
and possibly additional work, some of which
may be subject to precomputations.
\item The ``hypergeometric'' terms
asymptotically require \emph{less} than a full multiplication.
The precise amount of work varies depending on the formula.
\end{itemize}

\begin{table}
\setlength{\tabcolsep}{3pt}
\renewcommand{\arraystretch}{1.02}
\centering
\caption{Algorithms for the gamma function. Where no reference is given, the convergence rate
estimates will be justified below.}
\label{tab:algoverview}
\small
\begin{tabular}{p{4.1cm} p{4.5cm} l}
Formula                      & Needed terms for $p$-bit accuracy                    & Type of terms \\ \hline
    \multicolumn{3}{l}{\textbf{Global methods} (any $z$)} \\
Stirling series              & $0.323 p$ (worst case)        & Generic/hypergeometric \\
Spouge's formula       & $0.377 p$ (uniformly)         & Generic \\
                             & $< 0.225 p$ (small $z$, conjectured)    &  \\
Lanczos's formula      & unknown & Generic \\
Binet convergent series      & $1.00 p$ (worst case)        & Generic \\
Stieltjes continued fraction   & $0.323 p$ (worst case, conjectured)        & Generic \\
   \multicolumn{3}{l}{\textbf{Local methods} ($|z| \ll p$)} \\
Taylor series  & $O(p / \log(p))$             & Generic \\
Chebyshev interpolants  & $O(p / \log(p))$             & Generic \\
   \multicolumn{3}{l}{\textbf{Hypergeometric methods} ($|z| \ll p$)} \\
Incomplete gamma functions   & $1.49 p$        & Hypergeometric \\
Bessel functions             & $1.09 p$        & Hypergeometric \\
Elliptic integrals           & $O(\log p)$ (special $z$ only)        & Generic \\
   \multicolumn{3}{l}{\textbf{Integration methods}} \\
Contour integration          & $0.311 p$ (conjectured) \cite{schmelzer2007computing}  &  Generic \\
\end{tabular}
\end{table}

The performance of different algorithms
also sometimes depends strongly on whether $z$ is
an integer, rational, algebraic, real or complex number,
or a power series with the respective types of coefficients.
Despite these variables, the rate of convergence is useful as a first point
of reference. A more detailed comparison will be one
of the goals of this study.

Integration methods
are the only class that will not be studied here; the reader may
instead refer to Schmelzer and Trefethen \cite{schmelzer2007computing}.
Despite potentially very rapid convergence, numerical
integration is unlikely to be competitive
with more specialized algorithms
due to requiring repeated evaluations of a transcendental integrand.
Direct integration has one major advantage:
it is easy to generalize to other functions.
The integral \eqref{eq:gammadef} is the Laplace
transform of $t^{z-1}$ as well as the Mellin transform
of $e^{-t}$;
it is clearly interesting to have efficient schemes for
evaluation of similar integrals when no convenient closed form is available.

Notably missing from the list is \emph{solving an ordinary differential equation}.
Indeed, the gamma function itself does not satisfy any algebraic differential equation with polynomial
coefficients (H\"older's theorem~\cite{holder1886ueber}).
In the hypergeometric methods, the gamma function is effectively
viewed as the connection coefficient
between the formal solutions
of a differential equation at two singular points, but here $z$ appears as a parameter of the differential
equation, not as the integration variable.
The lack of a ``nice'' differential equation precludes asymptotically fast techniques available for many
other functions such as elementary functions, $\operatorname{erf}(z)$, Bessel functions $J_n(z)$, etc.\ (technically, \emph{holonomic} or \emph{D-finite} functions~\cite{vdH:hol,Mezzarobba2011}).

\subsection{Contents and contributions}

This work is structured as follows.
Section~\ref{sect:general} reviews notation,
formulas, precision issues and implementation
techniques which are common to all algorithms for the gamma function.

Section~\ref{sect:stirlingimpl} treats the Stirling series in depth.
We give pseudocode and discuss error bounds, precision loss,
parameter choices, convergence rates, and implementation techniques.
The most interesting contribution is an improvement to
the Stirling series (Theorem~\ref{thm:stirlexpand} and Algorithm~\ref{alg:stirlingsum})
which simultaneously improves performance and requires
fewer Bernoulli numbers
(this leads to the fastest
known general method for computing the gamma function
up to about $10^6$~digits).

Section~\ref{sect:global} discusses global alternatives to the
Stirling series. We analyze the efficiency of several algorithms
and attempt to answer whether the Stirling series should be dethroned. (Spoiler: it should not.)

Section~\ref{sect:local} discusses local methods.
Our main contribution is to analyze
computation using Taylor series. Our results include
a simple lower bound for the complex gamma function (Theorem~\ref{eq:globalbound})
which is new or at least does not appear in any of the usual reference works,
and new, improved bounds for Taylor coefficients
of the reciprocal gamma function (Theorem~\ref{eq:anbound} and formula \ref{eq:bestanbound}).

Section~\ref{sect:hypergeometric} discusses reduced-complexity
methods using hypergeometric series. We provide a comprehensive
collection of hypergeometric representations
of the gamma function and analyze their efficiency.
Section~\ref{sect:implresults} presents benchmark results
for multiple algorithms, and section~\ref{sect:openprob} discusses open problems.

Most of the algorithms
have been implemented in Arb~\cite{Joh2017}.
This study was prompted by a rewrite
of gamma and related functions in Arb undertaken by the author in 2021,
for which we wanted to explore plausible alternative algorithms
and new optimizations.
The work builds on, but significantly expands and updates, results reported in the thesis~\cite{Johansson2014thesis}.

\section{General techniques and preliminaries}

\label{sect:general}


\subsection{Variants of the gamma function}

A number of variants of $\Gamma(z)$ appear in applications
and may be provided as standalone functions in software:

\begin{itemize}
\item The factorial $z!$, which avoids a sometimes awkward shift by one.
\item The reciprocal gamma function $1 / \Gamma(z)$, which is an entire function and avoids issues with division by zero at the poles of $\Gamma(z)$ at $z = 0, -1, \ldots$.
This function is also useful in situations where $\Gamma(z)$ risks overflowing (where $1 / \Gamma(z)$ can evaluate to an upper bound or underflow to zero).
\item The scaled gamma function $\Gamma^{*}(z) = (2 \pi)^{-1/2} e^z {z}^{1/2-z} \Gamma(z)$, which by \eqref{eq:stirlingseries1} satisfies $\Gamma^{*}(z) \to 1$ when $z \to +\infty$, avoiding overflow and conditioning issues for large arguments \cite{gil2007numerical,Nemes2015}.
\item The log-gamma function $\log \Gamma(z)$, whose principal branch is defined to be holomorphic at $+\infty$ with branch cuts on the negative real axis, continuous from above.\footnote{We distinguish $\log \Gamma(z)$ from the pointwise composition $\log(\Gamma(z))$. Throughout this work, $\log(z)$ (with parentheses around the argument) denotes the principal branch of the natural logarithm, with continuity from above on the branch cut on $(-\infty, 0)$ so that $\log(-x) = \log(|x|) + \pi i$ for $x > 0$. In environments with signed zero, like IEEE~754 arithmetic, continuity on the branch cuts may follow the sign of the imaginary part, and algorithms have to be adapted accordingly.} The log-gamma function is useful for large arguments (avoiding overflow and conditioning issues) and for turning products or quotients of gamma functions into sums; some applications also require $\log \Gamma(z)$ itself.
\item The digamma function $\psi(z) = d/dz \, \log \Gamma(z) = \Gamma'(z) / \Gamma(z)$.
\item The derivatives $\Gamma^{(n)}(z)$, $[1/\Gamma(z)]^{(n)}$ and $\psi^{(n)}(z)$. The functions $\psi^{(n)}(z)$ are also known as polygamma functions.
\end{itemize}

Algorithms for the functions listed above are closely related
if not interchangeable.
In particular, any algorithm for $\Gamma(z)$ can be differentiated
to obtain the derivatives $\Gamma^{(n)}(z)$, $\psi^{(n)}(z)$, etc.
This is sometimes best done by differentiating formulas manually,
but in many cases, we can simply implement the original formulas
using automatic differentiation; that is, we take $z$
to be a truncated power series
\begin{equation}
z = c_0 + c_1 x + \ldots + c_n x^n \; \in \; \CC[[x]] \langle x^{n+1} \rangle.
\end{equation}
When used with fast power series arithmetic (based on FFT multiplication),
this also typically yields the best complexity for computing
high-order derivatives.

Algorithms for the gamma function also
permit the computation of rising factorials (Pochhammer symbols),
falling factorials,
binomial coefficients,
and the beta function.
Some techniques are applicable to
incomplete gamma and beta functions,
the Barnes $G$-function, the Riemann zeta function,
and other special functions.
To keep the scope manageable,
we will not consider such functions except where
they are involved in calculating the gamma function itself.

The gamma function can be defined for a matrix argument.
In general, the scalar algorithms can be generalized
to the matrix case; see~\cite{Cardoso2019,Miyajima2020}.
For the $p$-adic analog of the gamma function, see~\cite{villegas2007experimental}.

\subsection{Precision, accuracy and complexity}

Except where otherwise noted, $p$ either denotes a precision in bits
(meaning that real and complex numbers are represented by~$p$-bit
floating-point approximations) or a target accuracy (meaning
that we target a relative error of order $2^{-p}$).
For an introduction to arbitrary-precision arithmetic and many
of the underlying algorithms and implementation techniques,
see Brent and Zimmermann~\cite{mca}.

We typically need some extra bits of working precision for full $p$-bit accuracy,
where the precise amount depends on the input and
on the algorithm.
We will analyze sources of numerical error in some cases,
but we will not attempt to write down explicit bounds for errors
resulting from floating-point rounding.
Most algorithms can be implemented in ball arithmetic~\cite{vdH:ball,Joh2017}
to provide rigorous error bounds.
To compute the gamma function with a guaranteed relative error
$2^{-p}$, it thus suffices to choose
a working precision $p' > p$ heuristically;
we can restart with increased precision in the rare event that the
heuristic estimate turns out to be inadequate, and any complexity
analysis only depends on the estimate for $p'$ being accurate asymptotically.


To guarantee \emph{correct rounding},
Ziv's strategy~\cite{Ziv1991,MullerEtAl2018} should be used: we compute a first approximation
with a few guard bits and restart with increased precision if the
correctly rounded result cannot be determined unambiguously.
To ensure termination,
this loop needs to be combined with a test for input $z$ where $\Gamma(z)$
is exactly representable, which (in the case of floating-point input $z$) conjecturally occurs
only for $z \in \ZZ$.
For efficiency reasons, certain limiting cases should also
be handled specially (e.g.\ $\Gamma(z) \sim 1/z - \gamma$, $z \to 0$).

We will state some complexity bounds in the form $\bigOtilde(n)$,
meaning $O(n \log^c n)$ where we do not care about the constant $c$.
\emph{Arithmetic complexity} bounds assume that operations
have unit cost, while \emph{bit complexity} bounds account for the
cost of manipulating $p$-bit numbers in the multitape Turing model.
Additions and subtractions cost $O(p)$ bit operations
while multiplications cost $M(p) = O(p \log p) = \bigOtilde(p)$
with the asymptotically fastest known algorithm \cite{Harvey2021}.
Divisions and square roots cost $O(M(p))$
and elementary functions cost $O(M(p) \log p)$.
For moderate number of bits $p$ (up to a few thousand, say),
it is more accurate to estimate multiplications as having
cost $O(p^2)$ while elementary functions cost $O(p^{2.5})$.

We will sometimes use the fact that arithmetic operations and elementary
functions can be applied to power series of length $n$
using $O(n \log n) = \bigOtilde(n)$ arithmetic operations on coefficients.
This does not always translate to a uniform bit complexity bound,
because power series operations often lose $O(n)$ or $O(n \log n)$
bits of accuracy to cancellations (i.e.\ we will often need $p = \bigOtilde(n)$).

\subsection{Error propagation}

In ball arithmetic or interval arithmetic, the input $z$ to a function $f$ may be
inexact, i.e.\ it may be given by a rectangular enclosure
$z \in [m \pm r]$, $z \in [a \pm r] + [b \pm s] i$ or a complex ball $z = D(m, r)$.
We can let the uncertainty in the input propagate
automatically through an algorithm that computes $f$, but this sometimes yields
needlessly pessimistic enclosures.

A remedy is to evaluate $f$ at an exact
floating-point number $m \in E$, typically the midpoint of $E$, where $E$ is the enclosure for $z$.
We can then bound the propagated error using the radius of $E$
and an accurate bound for $\sup_{t \in E} |f'(t)|$.
Suitable bounds for derivatives of $\Gamma(z)$, $1/\Gamma(z)$ and $\log \Gamma(z)$
can be computed using the Stirling series or
other global approximations that will be discussed later (machine-precision
accuracy is sufficient for this).

When the input enclosures are wide, it is better to evaluate
lower and upper bounds.
The functions $\Gamma(x)$, $1/\Gamma(x)$, $\log \Gamma(x)$
are monotonic on the intervals $(0, x_0]$ and $[x_0, \infty)$
where $x_0 = 1.4616\ldots$ while the functions $\psi^{(m)}(z)$
are monotonic on $(0, \infty)$.
These conditions can be extended to segments on the negative real line
via the reflection formula.
In the complex domain, monotonicity conditions for
the real and imaginary parts of $\Gamma(z)$ are more complicated;
a convenient solution
is to exponentiate an enclosure of $\log \Gamma(z)$.

\subsection{Rising factorials}

\label{sect:rising}

The shift relation $\Gamma(z+1) = z \Gamma(z)$ can be used in two ways:
either to reduce the argument $z$ to a fixed strip near the origin,
for example $1 \le \operatorname{Re}(z) \le 2$ (useful for local and hypergeometric methods),
or to ensure that $|z|$ or $\operatorname{Re}(z)$ is large (for use with asymptotic methods).

Let us now consider the problem of evaluating a repeated argument
shift, i.e.\ computing the rising factorial
\begin{equation}
(z)_n \,=\, z (z+1) \cdots (z+n-1) \,=\, \frac{\Gamma(z+n)}{\Gamma(z)}, \quad n \ge 0.
\end{equation}

When $n$ is extremely large, the best way to compute $(z)_n$ numerically is to
use asymptotic formulas for the gamma function or for $(z)_n$ itself,
but our interest here is in smaller $n$ (of order $n \lesssim p^{1+\varepsilon}$) where we want to use rising factorials
in the computation of the gamma function and not vice versa.\footnote{Rising factorials
are also useful for computing binomial coefficients ${z \choose n} = (z-n+1)_n / n!$.}

The rising
factorial also serves as a model problem for the evaluation of hypergeometric series which will be considered later.
We recall that a sequence $f(n)$ is hypergeometric if it satisfies
a recurrence relation of the form $f(n+1) / f(n) = P(n) / Q(n)$
for some polynomials $P$, $Q$, or equivalently if $f(n)$
can be written as
\begin{equation}
f(n) = c d^n \frac{(a_1)_n \cdots (a_p)_n}{(b_1)_n \cdots (b_q)_n}.
\end{equation}
for some constants $a_1, \ldots, a_p, b_1, \ldots, b_q, c, d$.

\subsubsection{Binary splitting}

The obvious, iterative algorithm
to compute $(z)_n$ can in some circumstances
be improved using a divide-and-conquer strategy.

\begin{algorithm}
\caption{Rising factorial using binary splitting}\label{alg:rfbs}
\small
\begin{algorithmic}[1]
\Require $z \in R$, $n \in \NN$ where $R$ is any ring with $1 \in R$
\Ensure $z (z+1) \cdots (z+n-1) \in R$
\If {$n \le 10$} \Comment{Tuning parameter}
  \State \Return $z (z+1) \cdots (z+n-1)$ \Comment{Direct iterative product}
\EndIf
\State $A \gets (z)_m, \; B \gets (z + m)_{n - m}$, \quad $m = \lfloor n / 2 \rfloor$  \Comment{Recursive calls}
\State \Return $AB$
\end{algorithmic}
\end{algorithm}

Algorithm~\ref{alg:rfbs} is typically more numerically stable than
the iterative algorithm, particularly in rectangular complex interval
arithmetic. Additionally, when $R = \QQ$,
its bit complexity is only $\bigOtilde(n)$
compared to $\bigOtilde(n^2)$ for the iterative algorithm.
A similar speedup occurs in many other rings with coefficient growth.

When $z \in \QQ$, it is best to keep numerators and denominators
unreduced and only canonicalize the final result.
If we want $(z)_n$ as a numerical approximation
rather than an exact fraction, no GCDs are needed at all.
When $z \in \ZZ$ is a small integer; that is, when we are
essentially computing $n!$, the basic binary splitting scheme
can be improved by exploiting the prime factorization of factorials~\cite{Borwein1985,Lus2008}.

\subsubsection{Rectangular splitting}

The following algorithm~\cite{Johansson2014rectangular} is better for
$z \in R$ where a ``nonscalar'' (full) multiplication $xy$ with $x, y \in R$
has high, uniform cost while a ``scalar'' operation
such as $x + cy$ is cheap.
This assumption typically holds, for example, when
$z$ is a $p$-bit real or complex number and $c \in \ZZ$
with $p \gg \log_2 |c|$, when $z$ is a matrix,
or when $z$ is a truncated power series.

\begin{algorithm}
\caption{Rising factorial using rectangular splitting}\label{alg:rfrs}
\small
\begin{algorithmic}[1]
\Require $z \in R$, $n \in \NN$ where $R$ is any ring with $1 \in R$
\Ensure $z (z+1) \cdots (z+n-1) \in R$
\State Select tuning parameter $1 \le m \le \max(n,1)$ 
\State Compute the powers $z^2, z^3, \ldots, z^m$
\State $k \gets 0$
\State $r \gets 1$
\While {$k < n$}
  \State $\ell \gets \min(m, n-k)$
  \State Expand $f = (X+k) \cdots (X+k+\ell-1) = \sum_{i=0}^{\ell} f_i X^i \in \ZZ[X]$
  \State $t \gets f(z) = f_0 + f_1 z + \ldots + f_{\ell} z^{\ell}$ \Comment Using precomputed powers of $z$
  \State $r \gets r \cdot t$
  \State $k \gets k + m$
\EndWhile
\State \Return $r$
\end{algorithmic}
\end{algorithm}

Algorithm~\ref{alg:rfrs} performs
$O(n)$ scalar operations and $m + n / m$ full multiplications,
where $m$ is a tuning parameter.
The number of full multiplications is minimized if we take $m \approx \sqrt{n}$,
giving roughly $2 \sqrt{n}$ such multiplications.
However, the cost also depends on the size of the coefficients $f_i$
and on implementation details, so the optimal $m$
must be determined empirically.
For example, a carefully tuned implementation for $z \in \RR$ in Arb uses
$m = \lfloor \min(\sqrt{n}, 8 + 0.2 \max(0, p-4096)^{0.4}, 60) \rfloor$ for $n > 50$,
and $m = 1, 2, 4$ or $6$ for smaller $n$.

The polynomial $f$
can be constructed using repeated multiplication by linear polynomials,
resulting in a triangular scheme for the coefficients.
It can also be constructed using binary splitting (applied
to a rising factorial in the ring $\ZZ[X]$), but for the degrees
that will be encountered in gamma function computation,
the triangular scheme suffices (it should be implemented using machine words when $n$ and $m$ are small enough
that overflow cannot occur).

The scalar operations $\sum_k f_k z^k$ can be implemented
using optimized code for dot products~\cite{Johansson2019}.
If $z$ is a truncated power series, the inner operations
can be viewed as a matrix multiplication; the algorithm is related to the
Brent-Kung algorithm~\cite{BrentKung1978} for power series composition.


\subsubsection{Logarithmic rising factorial}

\begin{algorithm}
\caption{Logarithmic rising factorial (with correct branches)}\label{alg:logrf}
\small
\begin{algorithmic}[1]
\Require $z \in \CC$ with $\operatorname{Im}(z) > 0$, $n \in \NN$
\Ensure $\log \, (z)_{n} = \sum_{k=0}^{n-1} \log(z+k)$
\State $f \gets z (z+1) \cdots (z+n-1)$ \Comment{Using fast rising factorial algorithm}
\State If $n \le 1$, \Return $\log(f)$
\State $s \gets z$, \; $m \gets 0$   \Comment{Set $s$ to a floating-point approximation of $z$}
\For{$k \gets 1, 2, \ldots, n - 1$}
    \State $t \gets s \cdot (z+k)$ \Comment{Approximate floating-point product}
    \If{$\Imag(s) \ge 0$ and $\Imag(t) < 0$}
        \State $m \gets m + 2$
    \EndIf
    \State $s \gets t$
    \State $s \gets s 2^e$ so that $|s| \approx 1$ \Comment{Renormalize to avoid potential overflow}
\EndFor
\If{$\Real(s) < 0$} 
    \If{$\Imag(s) \ge 0$}
        \State $m \gets m + 1$
    \Else
        \State $m \gets m - 1$
    \EndIf
    \State \Return $\log(-f) + \pi i m$
\Else
    \State \Return $\log(f) + \pi i m$
\EndIf
\end{algorithmic}
\end{algorithm}

The log-gamma function satisfies $\log \Gamma(z+1) = \log(z) + \log \Gamma(z)$,
or $\log \Gamma(z+n) = \log \, (z)_n + \log \Gamma(z)$ where
\begin{equation}
\log \, (z)_n = \sum_{k=0}^{n-1} \log(z+k) = \log(|(z)_n|) + \sum_{k=0}^{n-1} \arg(z+k) i.
\end{equation}
Since $\log \, (z)_n \ne \log((z)_n)$ in general when $z$ is not a positive real number,
we have to be careful to compute the correct branch.

Hare~\cite{Hare1997} gives an algorithm for
computing $\log \, (z)_n$ using $n$ multiplications and a single logarithm
evaluation. This is significantly better than evaluating~$n$ logarithms.
Hare's algorithm computes the product $\prod_{k=0}^{n-1} (z+k)$ iteratively,
incrementing a phase correction by $2 \pi i$ every time the
imaginary part of the partial product changes sign from nonnegative to negative (assuming that $\operatorname{Im}(z) > 0$;
input in the lower half-plane can be handled via complex conjugation).

Hare's algorithm has two minor problems
which we address with the improved Algorithm~\ref{alg:logrf}.
First, Hare's original algorithm does not take advantage of the fast rising factorial
algorithms described earlier.
This does not matter in machine arithmetic,
but in arbitrary-precision arithmetic, it is better to compute
$(z)_n$ and separately determine the phase correction using Hare's algorithm
in machine-precision arithmetic.
Second, Hare's original algorithm involves an exact sign test which can fail in ball arithmetic.
Fortunately, we do not need to run Hare's algorithm in ball arithmetic;
since we essentially only need to determine $\operatorname{Im}(\log \, (z)_n)$ to within $\tfrac{1}{2} \pi$,
it suffices to apply it to a floating-point approximation of $z$,
at least as long as the input ball for $z$ is precise.\footnote{The relative error for a complex
multiplication performed the obvious way in $p$-bit floating-point arithmetic
is bounded by $\sqrt{5} \cdot 2^{-p}$~\cite{brent2007error}, so the
relative error due to multiplications in $s$ at the end of the main loop in Algorithm~\ref{alg:logrf}
is of order $(1 + \sqrt{5} \cdot 2^{-p})^n - 1$. Accounting for additions,
the initial error in $z$, and possible overflow/underflow, Algorithm~\ref{alg:logrf} is
\emph{very conservatively} valid at least for $n \le 10^6, |z| < 10^6, |\Imag(z)| > 10^{-6}$,
the input balls for $z$ and $\Imag(z)$ being accurate to at least 30 bits,
and with 53-bit machine arithmetic for the floating-point part.
If $z$ is too large, too close to the real line, or given by a too wide ball,
or if $n$ is too large, or if underflow or overflow is possible,
we may fall back to
computing $\sum_{k=0}^{n-1} \arg(z+k)$ more directly.}
In the final step, we compute the logarithm of either $(z)_n$ or $-(z)_n$
depending on the sign of the approximate product, in order to avoid
issues around the branch cut discontinuity.

For $z$ close to the real line, it may be useful to note
that $|\arg(z)| < \tfrac{1}{n} \pi$ implies that $\log \, (z)_n = \log((z)_n)$.

\subsubsection{Derivatives}

The derivatives $[(z)_n]^{(k)}$ for $0 \le k < d$ can be computed simultaneously
by evaluating $(z + x)_n$ in $z \in \CC[[x]] / \langle x^d \rangle$.
The special case $d = 2$ gives us the sum of reciprocals
\begin{equation}
\frac{1}{z} + \frac{1}{z+1} + \ldots + \frac{1}{z+n-1} = \frac{[(z)_n]'}{(z)_n},
\end{equation}
used in argument reduction for the digamma function,
while the higher logarithmic derivatives give the sums $\sum_{i=0}^{n-1} (z+i)^{-k}$
for polygamma functions of higher order.

The best algorithm depends on $n$ and $d$ as well as the precision $p$
and the representation of the argument $z$.
If $d$ is small, binary splitting or rectangular splitting
should be used depending on the bit size of $z$.
If $n$ and $d$ are both large, binary splitting should be used;
when $n = d = O(p)$, this leads to quasi-optimal
$\bigOtilde(p^2)$ bit complexity (assuming the use of fast power series arithmetic).

In some other cases, it is more efficient to form all the powers of $z$
appearing in the expanded polynomial $(z + x)_n \in \ZZ[z,x]$
and compute the coefficients (called \emph{skyscraper numbers}~\cite{khovanova2013skyscraper}) using recurrence relations.
Appropriate cutoffs between the algorithms must be determined empirically.\footnote{See the function \texttt{arb\_hypgeom\_rising\_ui\_jet} in Arb.}

\subsubsection{Asymptotically fast methods}

It is possible to compute rising factorials using
only $\bigOtilde(n^{1/2})$
arithmetic operations: assuming for simplicity that
$n = m^2$, the idea is to expand the polynomial
$f(x) = x (x+1) \cdots (x+m-1) \in \ZZ[x]$
and then evaluate $f(z), f(z+m), \ldots, f(z+m(m-1))$
simultaneously with a remainder tree,
using fast polynomial arithmetic~\cite{Strassen1976,ChudnovskyChudnovsky1988,Ziegler2005,BostanGaudrySchost2007}.
Such methods are useful in finite rings like $\ZZ / q \ZZ$,
but they are numerically unstable over~$\RR$,
forcing use of high intermediate precision~\cite{KohlerZiegler2008},
and they have not been observed to yield a speedup
over rectangular splitting for
numerical rising factorials with $n, p < 10^6$~\cite{Johansson2014rectangular}.

\subsection{Reflection formula}

Thanks to the reflection formula,
algorithms for the gamma function only need to treat $z$ in the half-plane $\operatorname{Re}(z) \ge \tfrac{1}{2}$,
or indeed any half-plane $\operatorname{Re}(z) \ge C$ when combined with
a shift.\footnote{However, the reflection formula should not be applied overzealously;
many formulas and algorithms work perfectly well for
$\operatorname{Re}(z) < \tfrac{1}{2}$,
without application of \eqref{eq:reflection},
at least as long as $z$ is not too close to the negative real line.
See the remarks in Hare~\cite{Hare1997} and the discussion of the Stirling series below.}

When evaluating the sine function in \eqref{eq:reflection}, we should
reduce $\operatorname{Re}(z)$ to $[-\tfrac{1}{2}, -\tfrac{1}{2}]$
before multiplying by $\pi$
to avoid unnecessary rounding errors
(in an implementation,
it is convenient to have an intrinsic \texttt{sinpi} function
for this purpose).
In the complex plane, say for $|\operatorname{Im}(z)| > 1$, the formula
\begin{equation}
\frac{1}{\sin(\pi z)} = 
\begin{cases}
\displaystyle{\frac{2 i e^{\pi i z}}{e^{2\pi i z} - 1}}, & \operatorname{Im}(z) > 1, \\
\displaystyle{\frac{-2 i e^{-\pi i z}}{e^{-2\pi i z} - 1}}, & \operatorname{Im}(z) < 1
\end{cases}
\end{equation}
is useful
to avoid overflow (in ball or interval arithmetic, this also reduces inflation of error bounds).
In an implementation of the reciprocal gamma function $1 / \Gamma(z)$,
the division-free formula
\begin{equation}
\frac{1}{\Gamma\!\left(z\right)} = \left(\tfrac{1}{\pi}\right) \sin\!\left(\pi z\right) \Gamma\!\left(1 - z\right),
\label{eq:reflectionrec}
\end{equation}
should be used instead of computing $\Gamma(z)$ via \eqref{eq:reflection} and then inverting.

The reflection formula for the principal branch of the log-gamma function
can be written as
\begin{equation}
\log \Gamma(z) = \log(\pi) - \log \sin(\pi z) - \log \Gamma(1-z).
\end{equation}

The \emph{log-sine function} with the correct
branch structure can be defined via
\begin{equation}
\log \sin(\pi z) = \int_{1/2}^z \pi \cot(\pi t) \, dt,
\end{equation}
where the path of integration is taken through the upper half plane
if $0 < \arg(z) \le \pi$ and through the lower half plane
if $-\pi < \arg(z) \le 0$.
The function $\log \sin(\pi z)$ is holomorphic
in the upper and lower half planes and
on the connecting real interval $(0,1)$.
It has branch cuts on the real intervals $(1, 2)$, $(2, 3)$, $\ldots$,
continuous from below, and on $(-1,0)$, $(-2,-1)$, $\ldots$, continuous from
above.
It coincides with $\log(\sin(\pi z))$
on the strip $-\tfrac{1}{2} < \operatorname{Re}(z) < \tfrac{3}{2}$.
In general,
\begin{equation}
\log \sin(\pi z) = \log(\sin(\pi(z-n))) \; \mp \; n \pi i, \quad n = \lfloor \operatorname{Re}(z) \rfloor
\label{eq:lsineq}
\end{equation}
where the negative sign is taken if $0 < \arg(z) \le \pi$
and the positive sign is taken otherwise.
To compute $\log \sin(\pi z)$, we can reduce
$z$ to the strip $0 \le \operatorname{Re}(z) < 1$
using \eqref{eq:lsineq}; on this strip, we may then use
\begin{equation}
\log(\sin(\pi z)) =
\begin{cases}
\log(\tfrac{1}{2}(1-e^{2i\pi z})) - i \pi (z-\tfrac{1}{2}), & \operatorname{Im}(z) > 1, \\
\log(\tfrac{1}{2}(1-e^{-2i\pi z})) + i \pi (z-\tfrac{1}{2}), & \operatorname{Im}(z) < 1
\end{cases}
\end{equation}
for $|\operatorname{Im}(z)| > 1$ to avoid large exponentials
that may cause spurious overflow or loss of accuracy
(if $|\operatorname{Im}(z)|$ is large, a \texttt{log1p} function is also useful here).

When implementing the reflection formula in ball or interval arithmetic
with an inexact input $z$,
it is useful to evaluate at the midpoint of $z$
and bound the propagated error using the derivative $[\log \sin(\pi z)]' = \pi \cot(\pi z)$,
so that the $n$ in \eqref{eq:lsineq} will be exact.
When the input ball or interval straddles a branch cut of $\log \sin(\pi z)$ itself, we need to compute the union of the
images on both sides.

The proofs of the above formulas just involve analytic continuation
with some bookkeeping for the placements of branch cuts.
See \cite[Proposition 3.1]{Hare1997} for an alternative form of the reflection formula.


The reflection formula for the digamma function reads
\begin{equation}
\psi\!\left(z\right) = \psi\!\left(1 - z\right) - \pi \cot\!\left(\pi z\right)
\label{eq:digammarefl}
\end{equation}
where the cotangent should be evaluated using
\begin{equation}
\cot(\pi z) =
\begin{cases}
i \left(\displaystyle \frac{2 e^{2\pi i z}}{e^{2\pi i z} - 1} - 1\right), & \operatorname{Im}(z) > 1, \\
-i \left(\displaystyle \frac{2 e^{-2\pi i z}}{e^{-2\pi i z} - 1} - 1\right), & \operatorname{Im}(z) < 1,
\end{cases}
\end{equation}
when not close to the real line.

For the higher derivatives, the reflection formula
\begin{equation}
\psi^{(m)}\!\left(1 - z\right) = {\left(-1\right)}^{m} \left(\psi^{(m)}\!\left(z\right) + \pi \frac{d^{m}}{{d z}^{m}} \cot\!\left(\pi z\right)\right)
\end{equation}
can be evaluated using recurrence relations;
even more simply (and asymptotically
more efficiently), we can evaluate \eqref{eq:digammarefl} using
power series arithmetic.
For evaluation of derivatives $\Gamma^{(m)}(z)$,
we can similarly apply power series arithmetic directly to~\eqref{eq:reflection};
there is no need to apply the reflection formula
at the level of the logarithmic derivatives.

\subsection{Branch correction}

\label{ref:branchcorrection}

We may sometimes want to compute $\log \Gamma(z)$ via the
ordinary gamma function,
in which case
we have $\log \Gamma(z) = \log(\Gamma(z)) + 2 \pi i k$
where the branch correction $k$ can be determined
from a low-precision approximation of $\log \Gamma(z)$.
It is useful to alternate between the formulas
\begin{align}
\log \Gamma(z) &= \log(\Gamma(z)) + 2 \pi i \left\lceil \frac{\Imag(\log \Gamma(z))}{2 \pi} - \frac{1}{2} \right\rceil \\
               &= \log(-\Gamma(z)) + \pi i \left(2 \left\lceil \frac{\Imag(\log \Gamma(z))}{2 \pi}\right\rceil - 1\right)
\end{align}
so that we avoid problems near the branch cut discontinuities of $\log(\Gamma(z))$.

Simply taking the leading term in the Stirling series
appears to be sufficient to determine the branch correction anywhere in the complex plane.
A numerical check suggests that for any $z = x + yi \in \CC \setminus \{0, -1, \ldots \}$,
\begin{equation}
\left|\operatorname{Im}(\log \Gamma(z)) - ((x-\tfrac{1}{2}) \operatorname{arg}(z)+ y(\log(|z|)-1)) \right| \; < \; \frac{\pi}{2},
\end{equation}
though we have not attempted to prove this inequality
(the error is smaller than $\tfrac{\pi}{4}$ for $\Real(z) > 0$ and
smaller than $0.08$ for $\Real(z) > \tfrac{1}{2}$; see also the
discussion of error bounds in section~\ref{sect:stirlbounds}).

In a narrow bounded region, say $[1,2] + [0, 10] i$, it is unnecessary
to evaluate transcendental functions; we can more
simply precompute a table of
linear approximations for the curves where
$\Imag(\log \Gamma(z)) = \pi i (k+\tfrac{1}{2})$.


\section{The Stirling series}

\label{sect:stirlingimpl}

The Stirling series \eqref{eq:stirlingseries1} is generally implemented in the logarithmic form\footnote{Some authors assign different names (De Moivre, Binet) to different versions
of the Stirling series, but we will not bother with this distinction. See \cite{Borwein2018,Corless2019}.}
\begin{equation}
\label{eq:stirlingseries}
    \log \Gamma(z) = \left(z-\tfrac{1}{2}\right)\log z - z +
          \frac{\log(2 \pi)}{2}
            + \sum_{n=1}^{N-1} T_n(z)
          + R_N(z)
\end{equation}
where the terms are given by
\begin{equation}
T_n(z) = \frac{B_{2n}}{2n(2n-1)z^{2n-1}}
\end{equation}
and the remainder term satisfies\footnote{For a proof using the Euler-Maclaurin formula, see Olver~\cite[Chapter 8]{Olver1997}.}
\begin{equation}
\label{eq:stirlingremainder}
    R_N(z) = \int_0^{\infty} \frac{B_{2N} - B_{2N}(\{x\})}{2N(x+z)^{2N}} dx, \quad z \not\in (-\infty, 0], \quad N \ge 1.
\end{equation}

Here, $B_n$ denotes a Bernoulli number, $B_n(x)$ denotes a Bernoulli polynomial,
\begin{equation}
\frac{t e^{xt}}{e^t-1} = \sum_{n=0}^\infty B_n(x) \frac{t^n}{n!}, \quad B_n = B_n(0),
\end{equation}
and $\{x\} = x - \lfloor x \rfloor$ denotes the fractional part of $x$.

Computing $\Gamma(z)$ as $\exp(\log \Gamma(z))$
via \eqref{eq:stirlingseries} instead of
the exponential form \eqref{eq:stirlingseries1} has multiple benefits:
the coefficients are simpler, the main sum only
contains alternating powers of $z$,
the error is easier to analyze,
and we do not need to worry about intermediate overflow or underflow
in the prefactors.
The expansion~\eqref{eq:stirlingseries} is also easily extended
to derivatives.

The computation of $\Gamma(z)$, $1/\Gamma(z)$, $\log \Gamma(z)$ or their derivatives
using the Stirling series can be
broken down into the following subproblems:

\begin{enumerate}
\item Choosing parameters (working precision; whether to use the reflection formula; shift $z \to z + r$; number of terms $N$); bounding $|R_N(z+r)|$.
\item Generating the Bernoulli numbers.
\item Evaluating the main sum.
\item Evaluating the surrounding factors and terms (rising factorials, logarithms, exponentials, trigonometric functions) to reconstruct the function value.
\end{enumerate}

We have already discussed the evaluation of rising factorials
and the reflection formula;
we will now study the remaining issues in detail.

First of all, we give a pseudocode implementation (Algorithm~\ref{alg:stirling}).
Assuming that all sub-functions ($\sin(\pi z)$, etc.) are evaluated accurately
to the indicated precision, Algorithm~\ref{alg:stirling}
computes $\Gamma(z)$, $1/\Gamma(z)$ or $\log \Gamma(z)$
with relative error less than $2^{-p}$, at least for \emph{most} input,
as verified by a heuristic analysis (see below)
and randomized testing.\footnote{The algorithm passes $10^5$ test cases with non-uniformly random $z$ and $2 \le p \le 2000$.}
Ball arithmetic provides certified results: the step
``add error bound'' incorporates the remainder term,
while bounds for rounding errors in all other operations will be
added automatically.
The reader who wishes to implement the algorithm in
plain floating-point arithmetic should carry out
a full error analysis in order to guarantee a rigorous result.

\begin{algorithm}
\caption{Gamma function using the Stirling series}\label{alg:stirling}
\small
\begin{algorithmic}[1]
\Require $z \in \CC$, precision $p \ge 2$
\Ensure $\Gamma(z)$, $1/\Gamma(z)$ or $\log \Gamma(z)$
\State If computing $\log \Gamma(z)$, and $z \approx 1$ or $z \approx 2$, use higher precision $p \gets p + e$ where $e = \lfloor \max(0, -\log_2(\min(|z-1|, |z-2|))) \rfloor$ (but fall back to a Taylor series approximation $\log \Gamma(1+z) = -\gamma z + \ldots$, $\log \Gamma(2+z) = (1-\gamma) z + \ldots$ if, for example, $e > p/2$)
\State \textbf{Part 1: parameter selection} (using machine-precision arithmetic)
\State $\beta \gets 0.2$     \Comment{Tuning parameter}
\State $\textit{reflect} \gets (\operatorname{Re}(z) < -5 \land |\Imag(z)| < \beta p)$ \Comment{Use reflection formula? Tuning parameter.}
\State If $\textit{reflect}$ set $z' \gets 1 - z$, else $z' \gets z$
\State $r \gets 0$; \textbf{while} $|\Imag(z)| < \beta p \land (\Real(z'+r) < 0 \lor |z'+r| < \beta p)$ \textbf{do} $r \gets r + 1$
\State $N \gets 1$
\While{$|R_N(z'+r)| > 2^{-p}$} \Comment{Use the minimum of the bounds \eqref{eq:stirbound1}, \eqref{eq:stirbound2}, \eqref{eq:stirbound3}}
    \State $N \gets N + 1$
    \State Backup termination test: depending on the tuning parameters and the numerical approximations used so far, this loop may not terminate. If the bound for $|R_{N}(z'+r)|$ is not decreasing, restart with $r \gets 2r, N \gets 1$ (or stop with the best $N$ found so far if less than $p$-bit accuracy is acceptable, e.g.\ in ball arithmetic with an inexact $z$).
\EndWhile
\State $p' \gets p + 5$; \quad $p'' \gets p' + \lfloor \log_2(\max(1, X \log(\max(1, Y)))) \rfloor$ where $X = r$, $Y = |z'| + r$,  when computing $\log \Gamma(z)$ and $X = Y = |z'|+r$ otherwise
\State \textbf{Part 2: evaluation} (using $p''$ bits of precision except where noted)
\State If $\textit{reflect}$ set $z' \gets 1 - z$, else $z' \gets z$
\State $s \gets \sum_{n=1}^{N-1} T_n(z'+r)$  \Comment{Main sum, using $p'$ bits of precision}
\State $t \gets (z'+r-\tfrac{1}{2}) \log(z'+r) - (z'+r) + \tfrac{1}{2}\log(2 \pi) + s$
\State $t \gets t + [\pm \varepsilon]$, $\varepsilon = |R_N(z'+r)|$ \Comment{Add error bound ($[\pm \varepsilon] + [\pm \varepsilon] i$ if $z$ is complex)}
\If{$\textit{reflect}$}
\State \Return $\Gamma(z) = \pi \exp(-t) (z')_r  (1 / \sin(\pi z))$
\State \Return $1/\Gamma(z) = \exp(t) \sin(\pi z) / ((z')_r \pi)$
\State \Return $\log \Gamma(z) = \log \; (z')_r - t - \log \sin(\pi z) + \log(\pi)$
\Else
\State \Return $\Gamma(z) = \exp(t) / (z)_r$
\State \Return $1/\Gamma(z) = \exp(-t) (z)_r$
\State \Return $\log \Gamma(z) = t - \log \; (z)_r$
\EndIf
\end{algorithmic}
\end{algorithm}

\subsection{Derivatives}

Algorithm~\ref{alg:stirling} can be generalized to a power series argument,
or adapted to compute $\psi(z)$.
We give the main formulas here.
Differentiating $m \ge 0$ times results in
\begin{equation}
[\log \Gamma(z)]^{(m)} = \frac{d^m}{dz^m} \left[\left(z-\tfrac{1}{2}\right)\log z - z + \frac{\log(2 \pi)}{2} \right] + \sum_{n=1}^{N-1} T^{(m)}_n(z) + R^{(m)}_N(z),
\end{equation}
\begin{equation}
T_n^{(m)}(z) = (-1)^m \frac{(2n-1)_{m} B_{2n}}{2n (2n-1) z^{2n+m-1}},
\end{equation}
\begin{equation}
R_N^{(m)}(z) = (-1)^m (2N)_{m} \int_0^{\infty} \frac{B_{2N} - B_{2N}(\{x\})}{2N(x+z)^{2N+m}} dx.
\end{equation}

Specializing to $m = 1$ gives the expansion for the digamma function
\begin{equation}
\label{eq:stirlingdigamma}
\psi(z) = \log(z) - \frac{1}{2z} - \sum_{n=1}^{N-1}  \frac{B_{2n}}{2n z^{2n}} + R'_N(z)
\end{equation}
and for $m \ge 2$ the polygamma functions
\begin{equation}
\label{eq:stirlingpolygamma}
\psi^{(m-1)}(z) = \frac{(-1)^m}{z^{m-1}} \left[ (m-2)! + \frac{(m-1)!}{2 z} + \sum_{n=1}^{N-1} \frac{(2n+1)_{m-2} B_{2n}}{z^{2n}}\right] + R^{(m)}_N(z).
\end{equation}

The derivatives can be evaluated individually or simultaneously as a power series.

\subsection{Error bounds}

\label{sect:stirlbounds}

It is convenient to bound the remainder term
as a multiple of the first omitted term; that is,
we consider the function $C_{N,m}(z)$ defined by
\begin{equation}
R_N^{(m)}(z) = C_{N,m}(z) \, T_N^{(m)}(z).
\end{equation}

When the argument is real and positive,
it is well known that the error in the Stirling series
is bounded in magnitude by the first omitted term.
In fact, the error has the same sign as the first omitted term:

\begin{theorem}
\label{thm:stirlingboundreal}
For $m \ge 0$, $N \ge 1$, and $x > 0$,
\begin{equation}
0 < C_{N,m}(x) < 1.
\end{equation}
\end{theorem}

For complex variables,
we have the following formula,
originally (with $m = 0$) due to Stieltjes.

\begin{theorem}
For $m \ge 0$, $N \ge 1$, and complex $z$ with $z \not \in (-\infty, 0]$,
\begin{equation}
|C_{N,m}(z)| \le \phi(z)^{2N+m}, \quad \phi(z) = \frac{1}{\cos(\tfrac{1}{2} \operatorname{arg}(z))}.
\label{eq:stirbound1}
\end{equation}
\end{theorem}

\begin{proof}
A proof with $m = 0$ is given by Olver~\cite[Chapter 8]{Olver1997}, which we here trivially generalize to $m \ge 0$ (Olver also
gives the case $m = 1$ in a subsequent exercise).
Since $|z+x|^{-1} \le\phi(z) (|z| + x)^{-1}$,
\begin{align*}
\left| R_N^{(m)}(z) \right| & \le \phi(z)^{2N+m} \int_0^{\infty} \frac{(2N)_{m} |B_{2N} - B_{2N}(\{x\})|}{2N {(|z| + x)}^{2N+m}} dx \\
                            & = \phi(z)^{2N+m} |R_N^{(m)}(|z|)| \\
                            & \le \phi(z)^{2N+m} |T_N^{(m)}(z)|
\end{align*}
where the last step uses Theorem~\ref{thm:stirlingboundreal}.
\end{proof}

Evaluating $\phi(z)$ does not require trigonometric functions.
If $z = x+yi$, we have
\begin{equation}
\phi(z) = \sqrt{1 + u^2}, \quad u = \frac{y}{|z| + x} = \frac{|z| - x}{y}
\end{equation}
where we should choose the expression for $u$ that avoids cancellation,
according to the sign of the real part $x$. It may be useful to note that
\begin{equation}
\phi(z) \le \begin{cases}
    1 & \operatorname{arg}(z) = 0 \\
    1.083 & |\operatorname{arg}(z)| \le \tfrac{1}{4} \pi \\
    1.415 & |\operatorname{arg}(z)| \le \tfrac{1}{2} \pi \\
    2.614 & |\operatorname{arg}(z)| \le \tfrac{3}{4} \pi.
\end{cases}
\end{equation}

The Stieltjes bound is convenient as it applies anywhere in the
cut complex plane, but it is not optimal. In fact, it overshoots the true
error exponentially with increased $N$ or $m$.
This is not a fatal problem for computations: using the reflection formula
so that $|\operatorname{arg}(z)| \le \tfrac{1}{2} \pi$,
the overshoot is at most of order $2^{N+m/2}$,
and this factor can be overcome with some extra argument reduction if needed.
Nevertheless, we can do much better with some case distinctions.

The $m = 0$ case of the following formula
is due to Brent~\cite[Theorem 1 and Corollary 1]{Brent2018} (improving on bounds by Spira~\cite{Spira1971} and Behnke and Sommer~\cite{behnke1962theorie}).

\begin{theorem}
For $m \ge 0$, $N \ge 1$, and complex $z$ with $\operatorname{Re}(z) \ge 0$, $z \ne 0$,
\begin{equation}
|C_{N,m}(z)| \le 1 + \sqrt{\pi (N + \tfrac{1}{2} m)}.
\label{eq:stirbound2}
\end{equation}
\end{theorem}

\begin{proof}
The proof boils down to bounding the integral $\int_0^{\infty} |z + x|^{-s} dx$ which
Brent does with $s = 2N$. The proof proceeds identically but with $s = 2N + m$.
\end{proof}

Brent's formula is clearly better than the Stieltjes bound for large $N$ unless $|\operatorname{arg}(z)|$ is very small.
In that case, we may use the even stronger Whittaker and Watson bound~\cite[Section 12.33]{WhittakerWatson1920}
\begin{equation}
|C_{N,0}(z)| \le 1, \quad |\operatorname{arg}(z)| \le \tfrac{1}{4} \pi.
\end{equation}
Another result due to Brent is \cite[Theorem 2]{Brent2018}
\begin{equation}
|C_{N,0}(z)| < 1.113, \quad |z| \ge N, \; \operatorname{Re}(z) \ge 0.
\end{equation}
We have not checked whether these results extend to $m > 0$, however.

For $z$ in the left half-plane, a bound due to Hare~\cite{Hare1997} is useful.
The downside of this bound is that involves $|\operatorname{Im}(z)|$ rather than $|z|$,
and therefore does not improve when shifting the argument.
The main point is that if $|\operatorname{Im}(z)|$ is already large,
we can use the Stirling series directly in the left half-plane
and get reasonable bounds without first applying the reflection formula.
We state Hare's bound in a simplified form and generalized to
allow $m \ge 0$.

\begin{theorem}
For $m \ge 0$, $N \ge 1$, and complex $z = x + yi$ with $|y| \ne 0$,
\begin{equation}
|R_N^{(m)}(z)| \le 4 \sqrt{\pi (N + \tfrac{1}{2} m)} \; |T_N^{(m)}(y)|
\label{eq:stirbound3}
\end{equation}
or equivalently,
\begin{equation}
|C_{N,m}(z)| \le 4 \sqrt{\pi (N + \tfrac{1}{2} m)} \left| \frac{z}{y} \right|^{2N+m-1}. 
\end{equation}
\end{theorem}

\begin{proof}
Hare obtains $|R_N(z)| \le \sqrt{\pi} \, \Gamma(N-\tfrac{1}{2}) |B_{2N}| / (\Gamma(N+1) |y|^{2N-1} )$; the simplified bound above
follows by the same calculation as in~\cite[Corollary 1]{Brent2018}
and the extension to $m \ge 0$ proceeds as in our
generalization of Brent's formula.
\end{proof}

For error bounds for the Stirling series in exponential
form, see Boyd~\cite{Boyd1994} and Nemes~\cite{Nemes2015}.

\subsection{Parameter selection and convergence analysis}

We assume for simplicity that $m = 0$
and that $\operatorname{Re}(z) \ge 0$,
in which case
$|R_N(z)|$ is essentially of the same order of magnitude as the first omitted
term in \eqref{eq:stirlingseries}.

We have $|R_N(z)| \to 0$ as $|z| \to \infty$ uniformly as long as 
$|\operatorname{arg}(z)| \le \pi - \delta$ for any $\delta > 0$,
but the Stirling series is divergent: $|R_N(z)| \to \infty$ as $N \to \infty$
for any fixed~$z$.
The basic strategy to compute the gamma function is therefore to write
\begin{equation}
\Gamma(z) = \frac{\Gamma(z+r)}{z (z+1) \cdots (z+r-1)}
\end{equation}
where $r \ge 0$ is chosen so that $|R_N(z+r)| < 2^{-p}$.

To study the asymptotic dependence between $p$, $N$ and $|z|$ or $|z+r|$,
we can estimate $|R_N(z)| \approx (2N)! (2 \pi)^{-2N} |z|^{1-2N} \approx (2N)^{2N} e^{-2N} (2 \pi)^{-2N} |z|^{-2N}$.
Equating this with $2^{-p}$ and solving for $N$ gives
\begin{equation}
N \approx -\frac{\log(2)}{2 W_{-1}\!\left(-\frac{\log(2)}{2 \pi e \beta}\right)} \, p, \quad \beta = \frac{|z|}{p}
\end{equation}
where $W_k(x)$ denotes the Lambert $W$-function~\cite{Corless1996}.
The $k = -1$ branch is used since we want the smallest real solution for $N$,
where $|R_N(z)|$ is decreasing (the principal $k = 0$ branch gives
the larger solution where $|R_N(z)|$ is increasing).
Since $W_{-1}(x)$ is real-valued only for $-e^{-1} \le x < 1$, we need
$|z| \ge p \log(2) / (2 \pi) \approx  0.110318 p$ to have a solution for $N$, i.e.\ to be able to
achieve $p$ bits of accuracy.\footnote{We have only solved for an approximation of the error bound,
but we do get the correct asymptotics $|z| \ge (\log(2) / (2 \pi) + o(1)) p$, which can be justified rigorously.}

To determine the shift $\Gamma(z) \to \Gamma(z+r)$ and number of terms $N$
to use for a given $z$ and a target precision $p$,
it is convenient choose $r$ using a condition
of the form
$|z+r| \ge \beta p$, where $\beta$ is a tuning parameter with $\beta > \log(2) / (2 \pi) \approx  0.110318$.
Once we have computed such an $r$, it is easy to find the minimal corresponding $N$ by a linear search.\footnote{In an arbitrary-precision implementation, this search can be done using machine
arithmetic, using logarithms or exponents of quantities where necessary to avoid underflow and overflow.
Bounds for $\log_2 |B_{2n}|$ can be stored in a lookup table for small $n$ and can be computed using asymptotic
estimates for large $n$ (see also the remarks in Section~\ref{sect:genbern}).
In fixed precision; suitable $r$ and $N$ can simply
be tabulated in advance. For example, $r = N = 8$
gives $\log \Gamma(z)$ with relative error smaller than $2^{-53}$
assuming that $\operatorname{Re}(z) > 0$.}

In the worst case ($z \approx 0$),
we need to compute a rising factorial of length $r \approx \beta p$.
The total number of terms or factors in the rising factorial and in the
Stirling series is then $\beta p + N$; this estimate of the work
cost is minimized when $\beta \approx 0.155665$, where $\beta p + N \approx 0.322797 p$.
In practice, the optimal~$\beta$ may be slightly larger
since the operations in the rising factorial are cheaper,
and a good $\beta$ should be determined empirically for a given implementation.
For example, Arb uses $\beta$ between $0.17$ and $0.24$ (varying slowly
with the precision); an older version used $\beta = 0.27$.

In the favorable case where $z$ is already large, the Stirling series
performs better. For example, if $|z| \approx p$, equivalent
to $\beta = 1$, we need no argument reduction
and only $N \approx 0.0728 p$ terms in the series.

Let us now consider numerical issues,
explaining the choice of working precision in Algorithm~\ref{alg:stirling}.
Computing $\Gamma(z)$ or equivalently $\Gamma(z+r)$ with relative error $2^{-p}$ entails evaluating $\log \Gamma(z+r)$ with \emph{absolute} error about $2^{-p}$.
We therefore need about $\log_2 |\log \Gamma(z+r)| \approx \log_2((|z|+r) \log(|z|+r))$ extra bits of precision for
the leading terms, while it is sufficient to use $p$ bits for the main sum.

To compute $\log \Gamma(z)$ with relative error $2^{-p}$, we do not need extra precision due to $z$ being large,
but we may need extra precision to compensate for cancellation in the argument reduction.
Generically this cancellation is of order
\begin{equation}
\log_2(|\log \Gamma(z+r) - \log \Gamma(z)|) = \log_2\!\left(\left|\int_{z}^{z+r} \!\!\psi(t) \, dt\right|\right) \approx \log_2(r \log(|z|+r)|)
\end{equation}
bits, except at the simple zeros at $z = 1$ and $z = 2$
where we have $\log_2(|z-1|)$ or $\log_2(|z-2|)$ extra bits of cancellation (accounted for at the start of Algorithm~\ref{alg:stirling}).

Similar estimates can be made for the functions $\psi^{(m)}(z)$.

\subsection{Generating Bernoulli numbers}

\label{sect:genbern}

Brent and Harvey \cite{Brent2013} discuss
several ways to compute the first $n$ Bernoulli numbers,
including simple recursive methods with $\bigOtilde(n^3)$
bit complexity
and asymptotically fast $\bigOtilde(n^2)$ algorithms
based on expanding a generating function such as
$\sum_{n=0}^{\infty} B_n x^n / n! = x / (e^x - 1)$
using fast power series or integer arithmetic.

In practice, a version of the classical
\emph{zeta function algorithm} \cite{Chowla1972,fillebrown1992faster,Harvey2010} performs even better.
The idea behind the zeta function algorithm is as follows:
if we approximate the Riemann zeta function
$\zeta(2n)$ to about $\log_2 |B_{2n}|$ bits of precision
(using either the Dirichlet series
$\zeta(s) = \sum_{k=1}^{\infty} k^{-s}$ or the Euler product
$\zeta(s) = \prod_p (1 - p^{-s})^{-1}$),
we can recover the numerator and denominator of
$B_{2n} = (-1)^{n+1}2(2n)! \zeta(2n) / (2\pi)^{2n}$
using the Von Staudt-Clausen theorem
\begin{equation}
\Big( B_{2n} + \sum_{q \in Q} \frac{1}{q} \Big) \in \ZZ, \quad Q = \{ q : q \text{ is prime}, \; (q-1) \, | \, 2n \}.
\end{equation}

Algorithm~\ref{alg:bernoulli} is a version of the
zeta algorithm adapted for multi-evaluation of Bernoulli numbers:
if we generate $B_n, B_{n-2}, \ldots$ in reverse order,
we can recycle the powers in the Dirichlet series
so that each new term only requires a cheap scalar multiplication (which GMP allows performing in-place).
This algorithm is due to Bloemen \cite{Bloemen2009} (with minor differences
to the present pseudocode).

\begin{algorithm}
\caption{Multi-evaluation of Bernoulli numbers}\label{alg:bernoulli}
\small
\begin{algorithmic}[1]
\Require $n \in \NN$, $n \ge 2$ even
\Ensure $B_n, B_{n-2}, \ldots, B_2$  (generated one by one) as exact, reduced fractions
\State $p \gets \lceil [((n+1) \log_2(n) - n \log_2(2 \pi e))] + 10 + 3 \log_2(n) \rceil$ \label{alg:line:compp}
\State $r \gets \lceil n / (2 \pi e) \rceil$, rounded up to an odd integer   \label{alg:line:compr}
\State $u \gets (2 \pi)^2$, $v \gets 2 (n!) / (2\pi)^n$ \Comment{Prefactors as $p$-bit floating-point numbers}
\State $t_k \gets \lfloor 2^p / k^n \rfloor, k = 3, 5, \ldots, r$ \Comment{$k^{-n}$ as $p$-bit fixed-point numbers (integers)}
\State Precompute primes $q \le n + 1$
\While{$n \ge 2$}
  \State $s \gets t_r + t_{r-2} + \ldots + t_5 + t_3$ \Comment{Exact sum, reverse order for efficiency}
  \State $s' \gets s 2^{-p}$ \Comment{Sum as a floating-point number}
  \State $B \gets  v + v (s' + (1 + s') / (2^n - 1))$ \Comment Approximate $|B_n|$ as a $p$-bit floating-point number (the inner operations before the last addition can optionally be done with only $p - n$ bits of precision)
  \State $a/b \gets \sum_{(q-1)|n} 1/q$ \Comment{Von Staudt-Clausen fraction}
  \State $N \gets \lfloor (-1)^{\lfloor n / 2 \rfloor + 1} B + a/b + 1/2 \rfloor$ \Comment{$B_n + a/b$ as an exact integer}
  \State \textbf{yield} $(N b - a) / b$ \Comment{Output $B_n$ (with reduced denominator $b$)}
  \State $n \gets n - 2$
  \State $v \gets u v / ((n+1)(n+2))$
  \State $t_k \gets k^2 t_k, \; k = 3, 5, \ldots, r$ \Comment{Exact multiplications}
  \If{$n \equiv 0 \pmod {64}$, $n \ne 0$} \Comment{Update precision, shorten sum}
    \State Compute new $p', r'$ as on lines \ref{alg:line:compp} and \ref{alg:line:compr}
    \State $t_k \gets \lfloor t_k / 2^{p - p'} \rfloor$, $k = 3, 5, \ldots, r'$, and discard the $t_k$ with $k > r'$
    \State Round $u, v$ to precision $p'$
    \State $p \gets p', \; r \gets r'$
  \EndIf
\EndWhile
\end{algorithmic}
\end{algorithm}

The zeta function algorithm has $\bigOtilde(n^3)$ bit complexity but
runs an order of magnitude faster than a well-optimized implementation of
the $\bigOtilde(n^2)$ power series method in FLINT~\cite{flint2021} for reasonable~$n$,
being 10 times faster
for $n = 10^4$ (1.2 seconds versus 12 seconds)
and 5 times faster for $n = 10^5$ (10 versus 50 minutes).

In an implementation of $\Gamma(z)$ that may be called
with different $p$, it is useful to maintain
a cache of Bernoulli number that gets extended gradually.
Algorithm~\ref{alg:bernoulli} generates
Bernoulli numbers in reverse order, but 
the method also works well for producing an infinite stream
in the forward order: we simply call
Algorithm~\ref{alg:bernoulli} with early termination
to create batches of, say, 100 new
entries each time. The initialization
overhead for each such batch is negligible.

\subsubsection{Detailed analysis}

The formula $((n+1) \log_2(n) - n \log_2(2 \pi e))$
is an accurate upper bound for $\log_2 |B_n|$ for even $n \ge 26$.
We need a few more bits to account for rounding errors and
the truncation error of the Dirichlet series.

We evaluate the Dirichlet series using fixed-point arithmetic for efficiency reasons;
other operations should use floating-point arithmetic (or floating-point ball arithmetic).
With ball arithmetic, it is easy to track the error terms
at runtime to certify that Algorithm~\ref{alg:bernoulli} outputs
the correct numerator.\footnote{We expect that the algorithm is correct with the stated
number of guard bits, making ball arithmetic unnecessary, but we do not attempt a proof here.}

The Dirichlet series truncation error is
bounded by $\sum_{k=r+2}^{\infty} k^{-n} \le (r+2)^{-n} + ((n-1) (r+2))^{1-n}$.
Comparing this with $2^{-p}$, we obtain $r \approx (2 \pi e)^{-1} n \approx 0.0585 n$.
Factoring out even powers leaves only $(4 \pi e)^{-1} n \approx 0.0293 n$ terms.

The remarkable efficiency of Bloemen's algorithm is explained by
three facts: the small constant factor $(4 \pi e)^{-1}$,
most terms being much smaller than $2^p$,
and the cheap scalar updates.
The number of terms can be reduced by a factor $\log(n)$ by factoring
out all composite indices from the Dirichlet series; that is,
using the Euler product (Fillebrown's algorithm~\cite{fillebrown1992faster});
this is more efficient for computing isolated Bernoulli numbers,
but the Dirichlet series is superior for multi-evaluation
precisely because it leaves only $O(1)$ nonscalar operations for each $B_n$.

\subsection{Evaluating the main sum}

Smith~\cite{Smith2001} has pointed out three improvements
over using Horner's rule to evaluate the sum
\begin{equation}
\sum_{n=1}^{N-1} T_n(z) = \sum_{n=1}^{N-1} \frac{B_{2n}}{2n(2n-1) z^{2n-1}}.
\end{equation}
First, since most terms make a small contribution to the final result,
the working precision should change gradually with the terms.
Second, rectangular splitting (or a transposed version, which
Smith calls \emph{concurrent summation)} should be used to take advantage
of the fact that the initial Bernoulli numbers are rational numbers
with small numerator and denominator.
Third, since the Bernoulli numbers near the tail have numerators that are
larger than the needed precision, it is sufficient to
approximate them numerically.

We can improve things further by breaking the sum into two parts
\begin{equation}
\sum_{n=1}^{N-1} T_n(z) = \sum_{n=1}^{M-1} T_n(z) + \sum_{n=M}^{N-1} t_n \zeta(2n), \quad t_n = \frac{(-1)^{n+1} 2 (2n-2)!}{(2 \pi)^{2n} z^{2n-1}}.
\label{eq:stirlingsplit}
\end{equation}

The first part will only involve smaller Bernoulli numbers,
for which rectangular splitting yields the biggest improvement,
and the Riemann zeta function values in the second part can be computed numerically.
However, we we will not compute the zeta values $\zeta(2n)$ explicitly;
instead, we expand them in terms of their Dirichlet series and change the order of summation to take
advantage of the fact that the factors $t_n$ as well as the
Dirichlet series terms $1/k^{2n}$ for consecutive $n$ form hypergeometric sequences.
This re-expansion increases the overall number of terms, but the new terms
are cheaper to evaluate since they are purely hypergeometric,
and we do not need to generate the numerical values of the corresponding Bernoulli numbers or zeta values at all
(the zeta values appear implicitly, baked into the hypergeometric
sums, with little evaluation overhead).

\begin{theorem}
\label{thm:stirlexpand}
Let $1 \le M \le N$, let $\{M_k\}_{k=1}^{K}$ be a list of integers with $N = M_1 \ge M_2 \ge \ldots \ge M_K = M$,
and define $w = 1/z^2$ and $u = -1/(2 \pi z)^2$.
Then the main sum in the Stirling series is given by
\begin{equation}
\sum_{n=1}^{N-1} T_n(z) \; = \; \frac{1}{z} \sum_{n=1}^{M-1} \frac{B_{2n}}{2n(2n-1)} w^n + S_2 + \varepsilon
\end{equation}
where
\begin{equation}
S_2 = - 2 (2M-2)! z u^M \sum_{k=1}^{K-1} \frac{1}{k^{2M}} \sum_{n=M}^{M_k - 1} (2M-1)_{2n-2M} \left(\frac{u}{k^2}\right)^{n-M},
\label{eq:s2rewritten}
\end{equation}
and
\begin{equation}
|\varepsilon| \le \sum_{k=2}^{K-1} \left(\max_{M_k \le n \le N - 1} |t_n| \right) \frac{N-M_k}{k^{2 M_k}} + \left(\max_{M \le n \le N - 1} |t_n| \right) (N-M) \zeta(2M,K).
\label{eq:epsbound}
\end{equation}
with $t_n$ defined as in \eqref{eq:stirlingsplit}.
\end{theorem}

\begin{proof}
Expanding the Riemann zeta function in \eqref{eq:stirlingsplit} using its Dirichlet series
gives
\begin{equation*}
\sum_{n=M}^{N-1} t_n \zeta(2n) =
\sum_{k=1}^{\infty} \sum_{n=M}^{N-1} \frac{t_n}{k^{2n}} =
\sum_{k=1}^{K-1} \sum_{n=M}^{M_k - 1} \frac{t_n}{k^{2n}} +
\sum_{k=2}^{K-1} \sum_{n=M_k}^{N - 1} \frac{t_n}{k^{2n}} +
\sum_{k=K}^{\infty} \sum_{n=M}^{N-1} \frac{t_n}{k^{2n}}.
\end{equation*}
Rewriting the first nested sum on the right-hand side
gives $S_2$, and bounding the last two nested sums gives
the bound for $|\varepsilon|$.
\end{proof}

We have made use of the Hurwitz zeta function
$\zeta(s,a) = \sum_{k=0}^{\infty} (k+a)^{-s}$.
We mention that \eqref{eq:epsbound} can be simplified
to the slightly weaker bound
\begin{equation}
|\varepsilon| \le \sum_{k=2}^{K} \left(\max_{M_k \le n \le N - 1} |t_n| \right) (N-M_k) \zeta(2 M_k, k).
\end{equation}


We give a complete implementation
of Theorem~\ref{thm:stirlexpand} in Algorithm~\ref{alg:stirlingsum},
which is an almost verbatim transcription of the method
as it is implemented in Arb.\footnote{A small optimization has been omitted from the pseudocode: in the leading sum, we can group the denominators of a few consecutive Bernoulli numbers to avoid some divisions.}

\begin{algorithm}
\caption{Main summation in the Stirling series}\label{alg:stirlingsum}
\small
\begin{algorithmic}[1]
\Require $z \in \CC \setminus \{ 0 \}$, $N \ge 1$, precision $p \ge 1$
\Ensure $\sum_{n=1}^{N-1} B_{2n} / (2n(2n-1) z^{2n-1})$ 
\State \textbf{Part 1: initialization} (Operations can use machine precision.)
\State Let $\{ b_n \}_{n=1}^{N-1}$ be a nonincreasing sequence of bounds for $\log_2(|B_{2n} / (2n(2n-1) z^{2n-1})|)$
\State $K \gets 2 \textbf{ if } p \le 1024 \textbf{ else } \min(4 + \lfloor 0.1 \sqrt{\max(p-4096, 0)} \rfloor, 100)$ \Comment{Tuning parameter}
\State Compute $\{ M_k \}_{k=1}^K$:
\State $M_1 \gets N$.
\For{$2 \le k \le K$}
\State Set $M_k \gets N$ and then decrement $M_k$ as long as the error stays small:
while $M_k > 2 \text{ and } b_{M_k - 1} - 2 (M_k - 1) \log_2(k) + \log_2(N - (M_k - 1)) < -p,$ do $M_k \gets M_k - 1$.
\EndFor
\For{$k \gets 2, 3, \ldots, K$}
\State $M_k \gets \min(M_k, M_{k-1})$ \Comment{Make nonincreasing}
\EndFor
\While{$K \ge 2$ and $M_K = M_{K-1}$}
\State $K \gets K - 1$  \Comment{Trim unneeded sums}
\EndWhile
\State $M \gets M_K$
\State $\varepsilon \gets (N-M) \zeta(2M, K) 2^{b_M} + \sum_{k=1}^{K-1} (N-M_k) 2^{b_{M_k}} k^{-2 M_k}$ \Comment{Error bound}
\State $m_1 \gets \max(1, \lfloor \sqrt{N-M} \rfloor), \; m_2 \gets \max(1, \lfloor\sqrt{M} \rfloor)$  \Comment{Tuning parameters}
\State \textbf{Part 2: trailing sum} (Using precision $p' = p + b_M$ except where otherwise noted.)
\State $u \gets - (2 \pi z)^{-2}$, compute table of $u^k$, $0 \le k \le m_1$
\State $S_3 \gets 0$
\For{$k_{odd} \gets 1; k_{odd} < K; k_{odd} \gets k_{odd} + 2$}
    \For {$k \gets k_{odd}; k < K; k \gets 2k$}
        \State Compute $v_k \gets u^k / k^{2j}$ for $0 \le j \le \min(M_k - M - 1, m_1)$. This can use the $u^k$ table when $k$ is odd; for successive even $k$, we only need to scale by powers of two.
        \State $S_4 \gets 0$
        \For {$n \gets M_k - 1; n \ge M; n \gets n - 1$}
            \State Do the following at precision $p + b_n$:
            \State \hskip1.0em $S_4 \gets 2n (2n-1) S_4 + v_{(n - M) \bmod m_1}$
            \State \hskip1.0em If $n - M \ne 0$ and $(n - M) \bmod m_1 = 0$, then $S_4 \gets v_{m_1} S_4$
        \EndFor
        \State $S_3 \gets S_3 + S_4 / k^{2M}$
    \EndFor
\EndFor
\State $S_3 \gets -2 (2M\!-\!2)! u^M z \, S_3$
\State $S_3 \gets S_3 + [\pm \, \varepsilon]$ ($z$ real) or $[\pm \, \varepsilon] + [\pm \, \varepsilon] i$ ($z$ complex) \Comment{Add error bound}
\State \textbf{Part 3: leading sum} (Using precision $p$ except where otherwise noted.)
\State $w \gets 1/z^2$, and compute table of $w^k$, $0 \le k \le m_2$
\State Compute (or read from cache) the Bernoulli numbers $B_0, B_2, \ldots, B_{2M-2}$
\State $S_2 \gets 0$
\For{$n \gets M - 1; n \ge 1; n \gets n - 1$}
    \State $P, Q \gets B_{2n}$ as an exact fraction
    \State Do the following at precision $p + b_n$:
    \State \hskip1.0em $S_2 \gets S_2 + P w^{(n-1) \bmod m_2} / (Q (2n (2n-1)))$
    \State \hskip1.0em If $n - 1 \ne 0$ and $(n - 1) \bmod m_1 = 0$, then $S_2 \gets w^{m_2} S_2$
\EndFor
\State \Return $S_2 / z + S_3$
\end{algorithmic}
\end{algorithm}

We start by choosing $K$.
Experiments with precision up to $p = 10^6$
suggest that it is optimal to choose $K$ with
$2 \le K \le 100$ and growing proportionally
to $\sqrt{p}$.

Given this tuning parameter, we can do a linear search
to find nearly minimal $M_1, M_2, \ldots, M_K$ such that
\eqref{eq:epsbound} is of order $2^{-p}$.\footnote{To bound the right-hand side of \eqref{eq:epsbound} numerically, we can use the inequality
$\zeta(s,a) \le a^{-s} + ((s-1) a^{s-1})^{-1}$, valid
for $s > 1$, $a > 0$.}
If necessary,
we adjust the sequence $M_k$ to be nonincreasing.
The remaining steps of the algorithm boil down to
performing rectangular splitting evaluation of sums, recycling terms,
and choosing the precision optimally to match the magnitudes of terms.

\subsubsection{Benchmark results}

As shown in Table~\ref{table:stirlingalgbench},
Algorithm~\ref{alg:stirlingsum} is roughly 1.5 times faster
than Horner's rule at high precision.\footnote{The speedup over Horner's rule is even bigger (more than a factor 2) when $z$ is complex. This additional speedup is simply due to the suboptimality
of Horner's rule for evaluating a polynomial with real
coefficients at a complex argument.}
The number of Bernoulli numbers is simultaneously
reduced by more than a factor 2,
reducing the precomputation time by almost a factor 10 (the memory usage
for storing Bernoulli numbers is also reduced by more than a factor 4).
We return to benchmarking the gamma function
computation as a whole in Section~\ref{sect:implresults}.

\begin{table}[h!]
\setlength{\tabcolsep}{3pt}
\renewcommand{\arraystretch}{1.1}
\centering
\caption{Time to evaluate the sum in the Stirling series at $d$ digits of precision
($p = d \log_2(10)$) with $N$ chosen to give $|R_N(z)| < 2^{-p}$, here with
$|z| / p \approx 0.27$.
\emph{Horner} is using Horner's rule with the only optimization of
varying the precision to match the magnitudes of the terms.
Algorithm~\ref{alg:stirlingsum} uses the internal parameters $K$ and $M$.
\emph{Eval time} is the time in seconds assuming that Bernoulli numbers are cached,
while \emph{Bernoulli} is the time to compute the required ($N$ or $M$) Bernoulli numbers.}
\label{table:stirlingalgbench}
\small
\begin{tabular}{l l l | l l | l l l l}
    &     &  &  \multicolumn{2}{|c|}{Horner} & \multicolumn{4}{c}{Algorithm~\ref{alg:stirlingsum}} \\
$d$ & $z$ & $N$ & Eval time & Bernoulli & $K$ & $M$ & Eval time & Bernoulli \\
\hline
$100$ & $89.1$ & $41$ & $0.000015$ & - & 2 & 30 & 0.000013 & - \\
$1000$ & $896.1$ & $391$ & $0.00045$ & 0.0032 & 4 & 238 & 0.00032 & 0.0015 \\
$10000$ & $8969.1$ & $3864$ & $0.071$ & 0.64 & 21 & 1678 & 0.047 & 0.080 \\
$100000$ & $89691.1$ & $38664$ & $17.3$ & 264.1 & 61 & 14219 & 10.9 & 19.5 \\
\end{tabular}
\end{table}

\subsubsection{Derivatives}

Algorithm~\ref{alg:stirlingsum} can be differentiated to compute $\psi(z)$,
$\psi'(z)$, etc.
To compute several derivatives at once, it is easier and
more efficient to generate
the terms $T_n^{(0)}(z)$ one by one
and then compute $z^{m+1} T_n^{(m+1)}(z)$ from $z^m T_n^{(m)}(z)$ using recurrence relations.
The complexity of computing $n$ derivatives is $\bigOtilde(p^2 n)$.

To compute a large number of derivatives to high precision,
the sum in the Stirling series should be implemented
using binary splitting for power series~\cite[Algorithm 4.6.1]{Johansson2014thesis}.
When computing $n$ derivatives to precision $p$ with $n = O(p)$,
this achieves quasi-optimal bit complexity $\bigOtilde(p^2) = \bigOtilde(n^2)$.


\section{Other global methods}

\label{sect:global}


The Stirling series
uses an asymptotic
expansion at $z = +\infty$ to correct the error in \emph{Stirling's formula}
\begin{equation}
z! \; \approx \; \sqrt{2 \pi} z^{z+1/2} e^{-z}.
\end{equation}

Equivalently, the Stirling series is the asymptotic power series expansion
of $\Gamma^{*}(z) = (2 \pi)^{-1/2} e^z {z}^{1/2-z} \Gamma(z) \approx 1$ (in exponential form)
or
$\log \Gamma^{*}(z) = \log \Gamma(z) - (z-\tfrac{1}{2}) \log(z) + z - \tfrac{1}{2} \log(2\pi) = R_1(z) \approx 1/(12 z)$ (in logarithmic form).
We may consider alternative methods to compute these functions.

\subsection{The formulas of Lanczos and Spouge}

We can write the combination of Stirling's formula with an $r$-fold shift as
\begin{equation}
\Gamma(z) \, \sim \, \frac{\sqrt{2 \pi} (z+r)^{z+r-\tfrac{1}{2}} e^{-z-r}}{z (z+1) \cdots (z+r-1)}
\label{eq:stirlcombined}
\end{equation}

In effect, this corrects for the
contribution of the first $r$ poles of the gamma function,
thus providing a good approximation for $\operatorname{Re}(z) > -r$,
whereas the normal Stirling formula ($r = 0$) only accounts for the
essential singularity at infinity.

The formulas of Lanczos and Spouge both have the form~\cite{Spouge1994,laurie2005,causley2021gamma}
\begin{equation}
\Gamma(z+1) \approx (z+r)^{z+1/2} e^{-z-r} \left[ \sqrt{2 \pi} + \sum_{n=1}^{N} \frac{c_n}{z+n} \right]
\label{eq:spouge}
\end{equation}
where $r$ is a free real parameter and
$c_n$ are some constants that depend on $r$.
We note that \eqref{eq:spouge} has similar structure to a partial fraction expansion of~\eqref{eq:stirlcombined}.

In Spouge's formula, $N = \lceil r \rceil - 1$ and
the coefficient $c_n$ is
the residue of the function $\Gamma(z+1) (z+r)^{-z-1/2} e^{z+r}$ at $z = -n$,
\begin{equation}
c_n = \frac{(-1)^{n+1} e^{r-n} (r-n)^{n-1/2}}{(n-1)!}, \quad 1 \le n \le N.
\end{equation}

Using the Cauchy integral formula and some lengthy calculations,
Spouge shows that the relative error $\varepsilon$ in the approximation \eqref{eq:spouge} satisfies~\cite[Theorem~1.3.1]{Spouge1994}
\begin{equation}
|\varepsilon| \le \frac{\sqrt{r}}{(2\pi)^{r+1/2}} \frac{1}{\Real(z+r)}, \quad \text{ for } r \ge 3, \quad \Real(z+r) > 0.
\label{eq:spougebound}
\end{equation}

Spouge's formula thus requires $N \approx p \log(2) / \log(2 \pi) \approx 0.377146 p$ terms
uniformly for $p$-bit accuracy,
with a weak improvement for larger $\Real(z)$.

Numerical experiments (see Table~\ref{tab:spougetab}) actually suggest that the
bound \eqref{eq:spougebound} is
conservative and that the true rate of convergence is
significantly better when $z$ is small,
with perhaps $N \approx 0.225 p$ being sufficient.
No proof of this empirical observation is known.\footnote{Other authors have already made this observation, for example~\cite{schmelzer2007computing}. Smith~\cite{smith2006gamma} claims to be able to prove that the error is roughly $2^N / N!$, but this is clearly incorrect.}

\begin{table}
\setlength{\tabcolsep}{3pt}
\renewcommand{\arraystretch}{1.1}
\centering
\caption{Actual relative error and the bound \eqref{eq:spougebound} for Spouge's formula with parameter $r$.}
\label{tab:spougetab}
\small
\begin{tabular}{l l l l l}
 $r$ & Error ($z = \pi$) & Bound ($z = \pi$) & Error ($z = 10^6+\pi$) & Bound ($z = 10^6+\pi$) \\ 
 \hline
 $10$ &  $5.0 \cdot 10^{-14}$ & $1.1 \cdot 10^{-9}$  & $6.3 \cdot 10^{-18}$ & $6.3 \cdot 10^{-14}$ \\
 $100$ &  $5.0 \cdot 10^{-131}$ & $5.9 \cdot 10^{-82}$  & $4.0 \cdot 10^{-97}$ & $6.1 \cdot 10^{-86}$ \\
 $1000$ &  $3.0 \cdot 10^{-1335}$ & $8.3 \cdot 10^{-801}$  & $3.2 \cdot 10^{-913}$ & $8.3 \cdot 10^{-804}$ \\
 $10000$ &  $1.4 \cdot 10^{-13405}$ & $6.3 \cdot 10^{-7985}$  & $4.5 \cdot 10^{-9099}$ & $8.3 \cdot 10^{-7987}$ \\
\end{tabular}
\end{table}

The coefficients $c_n$ corresponding to the Lanczos approximation~\cite{Lanczos1964}
are computed entirely differently, and
we will not reproduce the details here
since the formulas are much more complicated.
The Lanczos approximation appears to be more accurate
than the Spouge approximation with the same parameter~$r$,
but unfortunately, we do not have a formula
for bounding the error rigorously; the error in the Lanczos approximation must be estimated
empirically for a chosen parameter~$r$.
For an in-depth analysis, see Luke~\cite[p.\ 30]{luke1969special} and Pugh~\cite{pugh2004analysis}.

\subsubsection{Stirling versus Spouge}

The Stirling series requires about $0.322 p$ terms in the worst
case (counting both argument reduction and evaluation
of the main sum). The Spouge formula requires about $0.377 p$ terms,
or perhaps $0.225 p$ terms for small~$z$ assuming that a heuristic
error estimate is valid. These figures suggest that Spouge's method
might be competitive. However, there are several disadvantages:
\begin{itemize}
\item Each term in the Spouge sum requires a division,
or two multiplications if we rewrite the sum
as an expanded rational function.
The Stirling series
costs significantly less than one multiplication per term.
If we clear denominators, the denominator in the Spouge sum becomes
a rising factorial which can be evaluated more quickly,
but no such acceleration seems to be possible for the numerator.
\item The Spouge sum involves large alternating
terms, requiring higher intermediate precision.
\item Generating the Spouge coefficients for $p$-bit precision
costs $O(p \log p)$ multiplications. This is favorable compared
to computing Bernoulli numbers for the Stirling series in a naive way,
but it is no longer favorable when the
Stirling series and Bernoulli numbers are implemented efficiently.
\item The Spouge formula requires different sets of
coefficients for different~$p$,
whereas the Stirling series can reuse the same Bernoulli numbers.
\end{itemize}

There seems to be no reason to prefer
the Spouge formula over the Stirling series
unless one is interested specifically
in a compact implementation rather than efficiency (in a high-level language,
Spouge's formula can be implemented in a single line of code).
The Lanczos approximation has essentially the same disadvantages,
but with more complicated coefficients
and without the convenient error bound.

Spouge's formula was used to compute
the gamma function in an earlier version of MPFR~\cite{mpfralg},
which now however uses the Stirling series.
We will present benchmark results
for an Arb implementation below in Section~\ref{sect:implresults}.

The Lanczos approximation was popularized by \emph{Numerical Recipes}~\cite{press1989numerical}
and appears in some library implementations of the
gamma function (for instance in Boost~\cite{boost2021}).
It has to our knowledge never been used
in an arbitrary-precision implementation.
Even in machine precision with precomputed coefficients~$c_n$,
it is dubious whether the Lanczos formula has any advantage
over the Stirling series.
For example, with a parameter optimized for 53-bit accuracy, the Lanczos approximation
requires $N = 13$ terms~\cite{boost2021}; this is comparable to the number of terms (including
argument reduction steps) in the Stirling series in the worst case,
but significantly more expensive considering the use of divisions.

\subsection{Convergent series}

The divergent Stirling series can be replaced by a convergent series
if we allow more general expansions than power series in $z^{-1}$.
There are several formulas of this kind,
including Binet's rising factorial series \cite{binet1839memoire,vanmieghem2021binets} 
\begin{equation}
R_1(z) = \sum_{n=1}^{\infty} \frac{c_n}{(z+1)_{n}}, \quad c_n = \frac{1}{2n} \sum_{k=1}^n \frac{k|s(n, k)|}{(k + 1)(k + 2)}, \quad \Real(z) > 0,
\label{eq:binetseries}
\end{equation}
where $s(n,k)$ is a Stirling number of the first kind.

For large $z$, the terms in \eqref{eq:binetseries}
initially decay like $|z|^{-n}$, but this decay is eventually dominated
by the nearly factorial growth of the coefficients $c_n$, leading to
abysmally slow convergence if we take $N \to \infty$ terms while $z$ is fixed.
If we on the other hand increase $z$ and $N$ simultaneously as in the
implementation of the Stirling series, the convergence is quite rapid.
With $r = \alpha N$ steps of argument reduction
for some tuning parameter $\alpha$,
\begin{equation}
\frac{c_N}{(z+r)_{N}} \approx \frac{N! r!}{(N+r)!} \approx \frac{N^N r^r}{(N+r)^{N+r}} = \frac{1}{\gamma^N}
\end{equation}
where $\gamma = (\alpha+1)^{\alpha+1} / \alpha^\alpha$,
so that we need $N \approx p \log(2) / \log(\gamma)$ terms of the series
for $p$-bit accuracy.
Counting the combined number of terms $N + \alpha N = (1+\alpha) N$ in the series evaluation and the argument reduction,
the cost is minimized by $\alpha = 1$,
where we need $2 p \log(2)/\log(4) = p$ terms in total for $p$-bit accuracy.

We conclude that Binet's convergent series \eqref{eq:binetseries} is usable,
but less efficient than the Stirling series.

Another convergent expansion due to Binet is
\begin{equation}
R_1(z) = \sum_{n=1}^{\infty} \frac{n \zeta(n+1, z+1)}{2(n+1)(n+2)}, \quad \Real(z) > 0,
\end{equation}
and we mention the globally convergent Gudermann-Stieltjes series~\cite{stieltjes1889developpement,Barata2011}
\begin{equation}
R_1(z) = \sum_{n=0}^{\infty} \left[ (z+n+\tfrac{1}{2}) \log\!\left(\tfrac{z+n+1}{z+n}\right) - 1 \right], \quad z \ne \{0, -1, \ldots \},
\end{equation}
as well as Burnside's formula ($\Real(z) \ge \tfrac{1}{2}$)
\begin{equation}
\log \Gamma(z) = (z-\tfrac{1}{2}) \log(z-\tfrac{1}{2}) - (z-\tfrac{1}{2}) + \frac{\log(2\pi)}{2} -
\sum_{n=1}^{\infty} \frac{\zeta(2n, z)}{2^{2n+1} n (2n+1)}.
\end{equation}
Blagouchine~\cite{Blagouchine2016} gives more examples and references.
Unfortunately, the above expansions converge too slowly or have to complicated
terms (involving transcendental functions or nested sums) to
be of any practical interest for computations.

Higher-order remainder terms of the Stirling series can also
be re-expanded in terms of convergent series~\cite{Paris1992,paris2014comments,Nemes2015}.
For example, with $u = 2\pi i z$, Paris~\cite{paris2014comments} gives the formulas
\begin{align}
R_N(z) & =-\frac{\Gamma(2N-1)}{2\pi i} \sum_{k=1}^\infty \frac{1}{k}\{e^{uk} \Gamma(2-2N,uk)-e^{-uk} \Gamma(2-2N,-uk)\}  \label{eq:parisser1} \\
       &=\frac{2(-1)^{N-1}z}{(2\pi z)^{2N-2}}\sum_{k=1}^\infty\frac{1}{k^{2N-2}}\int_0^{\infty} \!\!\frac{t^{2N-2}e^{-t}}{t^2+4\pi^2k^2z^2}\,dt. \label{eq:parisser2}
\end{align}
It is interesting to compare these series with
Theorem~\eqref{thm:stirlexpand}, which for some $M \le N$ re-expands an approximation of $R_M(z)$ (namely $R_M(z) - R_N(z)$)
rather than $R_M(z)$ itself.
It is unclear whether the terms in
series like \eqref{eq:parisser1}, \eqref{eq:parisser2} can be evaluated efficiently enough to
offer any savings.

\subsection{Continued fractions}

The Stieltjes continued fraction for the gamma function~\cite{Char1980,cuyt2008handbook}
is the expansion
\begin{equation}
R_1(z) = \frac{a_1}{z + \displaystyle{\frac{a_2}{z + \ldots}}}, \quad a_1 = \tfrac{1}{12}, \; a_2 = \tfrac{1}{30}, \; a_3 = \tfrac{53}{210}, \ldots,
\label{eq:stieltjesfrac}
\end{equation}
which converges for any $\Real(z) > 0$.
The coefficients do not have a convenient closed form, but they are known to satisfy $a_n \sim n^2 / 16$.
We do not have an explicit error bound for \eqref{eq:stieltjesfrac}, though
it should be possible to derive such a bound using the general theory of continued fractions.

Like the Binet rising factorial series \eqref{eq:binetseries}, the convergence for fixed $z$ is too slow (of type $N^{-O(1)}$)
to form the basis of an arbitrary-precision algorithm; however, numerical experiments~\cite[Section 12.2]{cuyt2008handbook} suggest
that \eqref{eq:stieltjesfrac} converges about as fast as the Stirling series when increasing $\Real(z)$ and the number of terms $N$
simultaneously. The Stieltjes continued fraction also has the attractive feature of being much
more accurate than the truncated Stirling series when both $z$ and $N$ are small.

There are also drawbacks: even if the rate of convergence is comparable, evaluating $N$ terms of the
continued fraction is more expensive than evaluating
$N$ terms of the truncated power series \eqref{eq:stirlingseries},
and the coefficients $a_n$ are harder to compute than Bernoulli numbers.
For these reasons, \eqref{eq:stieltjesfrac} is not likely to be competitive with
the Stirling series for
arbitrary-precision computation.
With precomputed coefficients, it can be a useful alternative
at low precision (single or tens of digits).

\subsection{Other alternatives}

Stirling's formula
can be modified to improve accuracy for small~$z$.
For example, setting $r = 1/2$ in
\eqref{eq:spouge}
gives an approximation for $\Gamma(z+1)$ with relative error bounded by
$0.0534 (\operatorname{Re}(z) + \tfrac{1}{2})^{-1}$ in the right half-plane,
about half the maximum error of Stirling's formula~\cite[Theorem~1.3.2]{Spouge1994}.
An analysis by Pugh~\cite{pugh2004analysis} suggests that
$r = -W_{-1}(-1/\pi)/2 \approx 0.819264$ is near-optimal for
an approximation of this form.

Many authors~\cite{Nemes2010,Mortici2014,Chen2016,Wang2016,luschny2016,morris2020tweaking}
have proposed other modifications to the prefactor of the
Stirling series,
optionally followed by an asymptotic series
or continued fraction development.
These formulas sometimes offer better uniform accuracy than the Stirling series
when used with only the initial term or a small number of terms,
but probably offer no substantial advantages for arbitrary-precision computation.

\subsection{Verdict on alternative global methods}

The Stirling series
concentrates the accuracy around the
essential singularity at infinity,
and does so to the extreme extent of having zero radius of convergence
(viewed as a power series in $z^{-1}$).
It is something of a misconception that this property is a disadvantage;
on the contrary, it allows for rapid convergence in the
\emph{combined} limit $|z| \to \infty, N \to \infty$,
which extends globally thanks to the analytic continuation formula $\Gamma(z) = \Gamma(z+1) / z$.

With optimal argument reduction, the surveyed global methods
perform worse than or on par with the Stirling series when $z$ is large
and at best perform marginally better
when $z$ is small.
When $z$ is small, it is even better
to use local methods (as discussed in the next section)
which specifically take advantage of $z$ being small
without simultaneously trying to accomodate the behavior of $\Gamma(z)$ at infinity.


\section{Local methods}

\label{sect:local}

Local polynomial or rational function approximations
have been used to implement
the gamma function since the early days of electronic computers.\footnote{For example, Hastings~\cite{hastings1955approximations}
gives polynomials of degree 5 to 8 for approximating $\Gamma(1+x)$ on $0 \le x \le 1$.
Rice~\cite{Rice1964} considers rational approximations of $\Gamma(x)$ on $[2, 3]$
while Cody and Hillstrom~\cite{Cody1967}
give rational approximations for $\log \Gamma(x)$ on the intervals $[0.5, 1.5]$, $[1.5, 4]$ and $[4, 12]$.
Numerical coefficients of Taylor series are given in Abramowitz and Stegun~\cite{AbramowitzStegun1964}, Wrench~\cite{Wrench1968} and elsewhere.
Lookup tables which may be used for interpolation
are of course even older: according to Gourdon and Sebah~\cite{sebah2002introduction},
Gauss ``urged to his calculating prodigy student Nicolai (1793-1846) to compute tables of $\log(\Gamma(x))$ with twenty decimal places''.
In the 20th century, tables become widespread thanks to books such as Jahnke and Emde~\cite{jahnke1909funktionentafeln}.}
They are an excellent complement to the Stirling series for the following reasons:

\begin{itemize}
\item We can avoid (some) argument reduction for small $z$.
\item We avoid one or two elementary function evaluations (depending on whether we compute $\log \Gamma(z)$ or $\Gamma(z)$ and whether argument reduction is needed).
\item With carefully chosen approximants, we can guarantee certain properties (correct rounding, monotonicity) at least near certain points.
\end{itemize}

Most modern machine-precision mathematical standard library (\texttt{libm}) implementations
combine the Stirling series with
polynomial or rational approximations on one or several short intervals,
although the choices of intervals, approximants and asymptotic cutoffs
vary between implementations.

Just to take one example,
OpenLibm~\cite{OpenLibm2021} implements the \texttt{lgamma} function ($\log \Gamma(x)$) using
the Stirling series for $x \ge 8$; for $0 < x < 8$,
it reduces the argument to the interval $[2, 3]$
where it uses a rational approximation of the form $\log \Gamma(x) \approx 0.5 s + s P(s)/Q(s), s = x - 2$,
with the exception of the interval $1.23164 < x < 1.73163$ where it uses
a degree-14 polynomial approximation chosen
to maintain monotonicity
around the minimum at $x \approx 1.46163$.

We will not attempt to reverse-engineer such fixed parameter choices here.\footnote{See Beebe~\cite{beebe2017mathematical} for a discussion
of implementation techniques emphasizing machine precision,
and Zimmermann~\cite{Zimmermann2021}
for a comparison of the accuracy of several \texttt{libm} implementations.}
We will instead study local methods for the purposes of arbitrary-precision computation.
In this setting, it makes sense to focus on approximations that are economical
in terms of precomputation time and storage,
and we will therefore focus on the use of
a single approximating polynomial.\footnote{We can clearly achieve arbitrary performance goals
if we drop this constraint, but enormous lookup tables have their own drawbacks.
Even machine-precision implementations tend to be frugal about lookup tables
in order to minimize code size and avoid cache misses. In IEEE 754 \texttt{binary64} arithmetic,
$\Gamma(x)$ could in principle be covered entirely with low-degree piecewise polynomial
approximations on the part $0 < x \lesssim 171.63$ of its positive domain
where it is finite, but we are not aware of any mainstream \texttt{libm} implementation that follows this approach.}

\subsection{Taylor series}
\label{sect:taylor}

A natural way to compute $\Gamma(z)$, $1/\Gamma(z)$, $\log \Gamma(z)$ or $\psi(z)$
on an interval $[a,b]$ is to
use the Taylor series or Laurent series centered on $z = (a+b)/2$~\cite{Wrench1968,Fransen1980}.\footnote{For $\log \Gamma(z)$, the expansions at $z = -n$ are not Laurent series but generalized series expansions with leading
logarithms.}
As discussed previously, the coefficients
at an arbitrary expansion point are easily computed using
the Stirling series together with power series arithmetic.

The Taylor series at $z = 0$ of the reciprocal gamma function
\begin{equation}
\frac{1}{\Gamma(z)} = \sum_{n=1}^{\infty} a_n z^n = z + \gamma z^2 + \left(\frac{\gamma^2}{2} - \frac{\pi^2}{12}\right)z^3 + \ldots
\label{eq:rgammaseries}
\end{equation}
is particularly interesting:

\begin{itemize}
\item The function $1/\Gamma(z)$ is an entire function, so its Taylor series converges
more quickly than the Taylor or Laurent series of $\Gamma(z)$ or $\log \Gamma(z)$.
Indeed, on any bounded domain, it suffices to take $O(p / \log p)$ terms
of \eqref{eq:rgammaseries} for $p$-bit accuracy, compared to $O(p)$ terms for expansions of $\Gamma(z)$.
\item The coefficients have a special form (see below) and can therefore potentially be computed more quickly than the coefficients at a generic point.
\item It is useful to choose an integer as the expansion point to optimize for input $z = n + \varepsilon$ close to integers. Such input often appears as a result of numerical limit computations or when solving perturbation problems.
\end{itemize}

The coefficients in \eqref{eq:rgammaseries} can be computed from the
the Taylor series
\begin{equation}
\log \Gamma (1-z) = \gamma z + \sum_{n=2}^{\infty} \frac{\zeta(n)}{n} z^n,
\end{equation}
requiring the values $\gamma, \zeta(2), \zeta(3), \ldots$ and
exponentiation of a power series. Explicitly, this leads to
the recurrence
\begin{equation}
(n-1) a_n = \gamma a_{n-1} - \zeta(2) c_{n-2} + \zeta(3) c_{n-3} - \ldots + (-1)^{n} \zeta(n-1) a_1
\label{eq:taylorrecurrence}
\end{equation}
for $n \ge 2$~\cite{Wrench1968}. However,
Newton iteration should be used instead of the direct recurrence \eqref{eq:taylorrecurrence}
when computing a large number
of coefficients (say, $N > 500$) in order to achieve a softly optimal
$\bigOtilde(N^2)$ complexity (assuming that $p = O(N)$).
See \cite[Section 4.7]{Johansson2014thesis} for a detailed
review of algorithms for computing $a_1, \ldots, a_N$.

When implementing $\eqref{eq:rgammaseries}$ to evaluate $1/\Gamma(z)$ or $\Gamma(z)$,
it is convenient to shift the expansion point by one, giving
\begin{equation}
\frac{1}{\Gamma(1+u)} = \sum_{n=0}^{\infty} a_{n+1} u^n = 1 + \gamma u + \ldots, \quad u = z - 1
\end{equation}
for use on the interval $z \in [0.5, 1.5]$,
or more generally on the strip $\Real(z) \in [0.5, 1.5]$.

To evaluate $\log \Gamma(z)$, we need a logarithm evaluation,
together with a branch correction (as discussed in section~\ref{ref:branchcorrection}) valid on $\Real(z) \in [0.5, 1.5]$
when $z$ is complex.

\subsubsection{Coefficient bounds and convergence}

The coefficients $a_n$ in \eqref{eq:rgammaseries} can be estimated accurately using
saddle-point analysis.
The best results of this type are due to Fekih-Ahmed~\cite{fekihahmed2017power}.
However, there are few published explicit upper bounds for $|a_n|$; we are only aware of Bourguet's bound \cite{Bourguet1883,fekihahmed2017power}
\begin{equation}
|a_n| \le \frac{e}{\pi^{n} (n+1)!} + \frac{4}{\pi^2 \sqrt{n!}} \lesssim \frac{4}{\pi^2 \sqrt{n!}}
\label{eq:bourguet}
\end{equation}
which is acceptable for $n < 100$ but quite pessimistic asymptotically.

We can get better
bounds with the help of the following global inequality for the reciprocal gamma function.

\begin{theorem}
\label{eq:globalbound}
For any $z \in \CC$,
\begin{equation}
\left| \frac{1}{\Gamma(z)} \right| \; \le \; \left| \sqrt{z} e^z z^{-z} \right| \; \le \; e^{\pi R/2} R^{1/2+R}, \quad R = |z|.
\end{equation}
\end{theorem}

\begin{proof}
We sketch the proof of the first inequality.
An explicit computation in interval arithmetic establishes the result for small $|z|$.
For large $|z|$ in the right half-plane, the inequality follows
immediately from the Stirling series.
In the left half-plane, the reflection formula together with the Stirling series
gives
\begin{equation}
\frac{1}{\Gamma(z)} = \sqrt{\frac{2}{\pi}} \sin(\pi z) \sqrt{z} e^{z} (-z)^{-z} (1 + \varepsilon)
\end{equation}
where $\varepsilon \to 0$ when $|z| \to \infty$. We can verify
$|(2 / \pi)^{1/2} (1 + \varepsilon)| \le 1$ for sufficiently large $|z|$
using the error bound for the Stirling series.
It remains to observe that $|(-z)^{-z} \sin(\pi z)| = |z^{-z} \sin(\pi z) e^{-\pi |\Imag(z)|}| \le |z^{-z}|$.

For the second inequality, let $z = R e^{i \theta}$ with $\theta \in (-\pi, \pi]$. Then
\begin{equation}
\left| \sqrt{z} e^z z^{-z} \right| = e^{R (\cos(\theta) + \theta \sin(\theta))} R^{1/2 - R \cos(\theta)}
\end{equation}
and the result follows from the fact that $|\cos(\theta) + \theta \sin(\theta)| \le \pi / 2$.
\end{proof}

An application of Cauchy's integral formula now yields
the following corollary.
\begin{theorem}
For $n \in \NN$ and any $R > 0$, $|a_n| \le e^{\pi R/2} R^{1/2+R-n}$.
\label{eq:anbound}
\end{theorem}

\begin{table}
\setlength{\tabcolsep}{3pt}
\renewcommand{\arraystretch}{1.05}
\centering
\caption{Taylor coefficients $a_n$ for the reciprocal gamma function $1 / \Gamma(z)$ at $z = 0$, and corresponding estimates and bounds.}
\label{table:tab1}
\small
\begin{tabular}{l l l l l l}
 $n$ & $a_n$ & Fekih-Ahmed  & Bourguet \eqref{eq:bourguet} & Bound \eqref{eq:bestanbound} & ($R = n/8$) \\ 
 \hline
$1   $ & $+1.0000$                    & $+0.98$                    & $4.68$ & $2.14$                   & $2.66$ \\
$10  $ & $-2.1524 \cdot 10^{-4}$      & $-2.01 \cdot 10^{-4}$      & $6.60 \cdot 10^{-3}$ & $8.13 \cdot 10^{-2}$     & $1.14$ \\
$10^2$ & $+6.6158 \cdot 10^{-106}$    & $+6.60 \cdot 10^{-106}$    & $4.20 \cdot 10^{-80}$ & $9.38 \cdot 10^{-91}$    & $1.25 \cdot 10^{-87}$ \\
$10^3$ & $+5.3533 \cdot 10^{-1871}$   & $+5.35 \cdot 10^{-1871}$   & $6.39 \cdot 10^{-1285}$ & $2.54 \cdot 10^{-1750}$  & $3.36 \cdot 10^{-1749}$ \\
$10^4$ & $+1.5010 \cdot 10^{-27327}$\!  & $+1.50 \cdot 10^{-27327}$  & $7.60 \cdot 10^{-17831}$ & $1.46 \cdot 10^{-26322}$ & $2.10 \cdot 10^{-26244}$ \\
$10^5$ &                            & $+7.10 \cdot 10^{-362317}$\! & $7.63 \cdot 10^{-228288}$\! & $4.23 \cdot 10^{-353745}$\! & $5.96 \cdot 10^{-349951}$\! \\ [1ex] 
\end{tabular}
\end{table}

We are free to choose $R$ as a function of $n$,
and we may in particular pick the optimal value
\begin{equation}
|a_n| \le e^{\pi R/2} R^{1/2+R-n}, \quad R = \frac{n-\tfrac{1}{2}}{W_0((n+\tfrac{1}{2}) e^{\pi/2+1})}.
\label{eq:bestanbound}
\end{equation}

For $n$ in a range relevant for computations,
$R = n/8$ is nearly as accurate (see Table~\ref{table:tab1}),
and gives tail bounds that are easy to compute.

In an implementation, we may want to determine
a tight bound by inspecting the computed coefficients.
An exhaustive computation
(checking worst cases of the possible term ratios $|a_{n+1}| / |a_n|$,
together with \eqref{eq:bestanbound} for a rough tail bound)
establishes the following:

\begin{theorem}
Let $b_n = a_{n+1}$. If $|z| \le 20$ and $N \le 1000$, then
\begin{equation}
\left| \frac{1}{\Gamma(1+z)} - \sum_{n=0}^{N-1} b_n z^n \right| \le 8 \max(\tfrac{1}{2}, |z|) |b_N| |z|^N
\end{equation}
provided that the right-hand side is smaller than $2^{-8}$.
The same statement also holds for $N \le 10000$ with the exception of $N \in \{1443, 2005, 9891\}$.
\end{theorem}

For any $C \ge 0$, $\lim_{p \to \infty} |C^N a_N| \to 0$ with $N = \lfloor p / \log(p) \rfloor$,
so it suffices to take $O(p / \log(p))$ terms
to compute $1/\Gamma(z)$ to $p$-bit accuracy for any fixed~$z$.\footnote{In fact, it appears that we can take
$N = o(1) \cdot p / \log(p)$, but $N = p / (\log(p) \log(\log(p)))$ is not quite sufficient.}
The cost of using the Taylor series grows roughly linearly with $|\Real(z)|$ due
to argument reduction.

With increased $|\Imag(z)|$, the cost grows as more terms are required.
Increased working precision is also required to counteract
cancellation (the precision must increase superlinearly with $|\Imag(z)|$
since $|1 / \Gamma(yi)|$ grows like $e^{\pi |y|/2}$
while $|1 / \Gamma(-|y|)|$ grows like $\Gamma(|y|)$).
It is possible to use the multiplication theorem
\begin{equation}
\prod_{k=0}^{m - 1} \Gamma\!\left(\frac{z + k}{m}\right) = {\left(2 \pi\right)}^{\left( m - 1 \right) / 2} {m}^{1 / 2 - z} \, \Gamma\!\left(z\right)
\end{equation}
to replace one evaluation far away from the real
line with $m$ evaluations near the real line.
This might not be faster than using the Taylor series once,
but it can avoid cancellation and reduce the
number of required Taylor coefficients.

\subsubsection{Taylor versus Stirling}

\begin{figure}
\caption{Region where the Taylor series theoretically is faster than the Stirling series, counting only arithmetic operations in argument reduction and polynomial evaluation (here for $p = 1000$). Observe that the scale of the imaginary axis is exaggerated in this visualization. \label{fig:taylordiamond}}
\includegraphics[width=10cm]{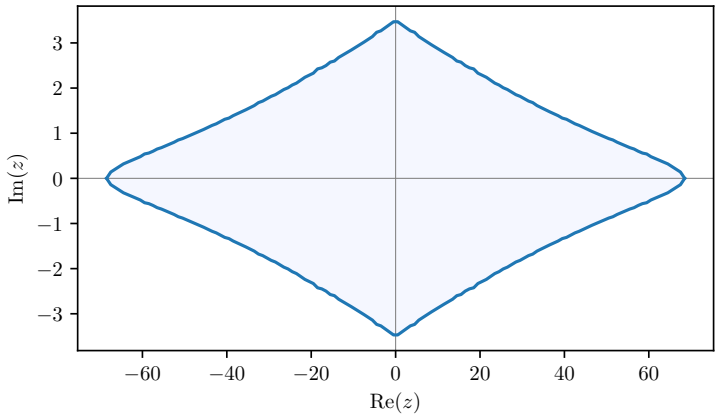}
\end{figure}

Asymptotically, the Taylor series should be faster than the Stirling series
inside a narrow diamond-shaped region (see Figure~\ref{fig:taylordiamond})
whose boundaries grow with the precision $p$.

We have not attempted to determine a formula for the asymptotic shape of this region.
In practice, the scale and shape of the boundary will be different from
a theoretical prediction
due to elementary function costs, non-constant costs of arithmetic operations, differences in working precision, etc.,
and suitable regions for using the Taylor series must be determined empirically.

\subsection{McCullagh's series}

Differentiating the Taylor series \eqref{eq:rgammaseries} allows
computing $[1/\Gamma(z)]'$ about as rapidly as $1/\Gamma(z)$,
and with two Taylor series evaluations, we can recover $\Gamma'(z)$ or $\psi(z)$.

An alternative expansion for $\psi(z)$ is McCullagh's series~\cite{McCullagh1981} \cite[Section~18.2.7]{beebe2017mathematical}
\begin{equation}
\psi(1+z) = -\gamma - \sum_{n=1}^{\infty} (-z)^n \left( d_n + \frac{c_n}{z+n} \right),
\end{equation}
where the coefficients
\begin{equation}
c_n = \frac{1}{n^n}, \; d_n = \sum_{k=n+1}^{\infty} \frac{1}{k^{n+1}} = \zeta(n+1, n+1).
\end{equation}
can be precomputed.
This series converges even more quickly
than the Taylor series for $1/\Gamma(z)$,
having terms of order $|z|^n n^{-n}$,
though it requires more work per term.

McCullagh's series can also be differentiated to obtain
expansions of $\psi^{(k)}(z)$.
Integrating term by term yields the rapidly convergent
\begin{align}
\log \Gamma(1+z) & = -\gamma z + \sum_{n=1}^{\infty} \frac{(-z)^{n+1}}{n+1} \left( d_n + \frac{c_n}{n} {}_2F_1\!\left(1, n+1; \, n+2; \, -\frac{z}{n}\right)\right) \\
                 & = -\gamma z + \sum_{n=1}^{\infty} \left[ \frac{(-z)^{n+1} d_n}{n+1} - \log\left(1+\frac{z}{n}\right) - \sum_{k=1}^{n} \frac{(-z)^k}{k n^k} \right]
\end{align}
which however is less suitable for computations since there does not seem to
be a simple recurrence relation for the logarithmic (${}_2F_1$ function) terms.








\subsection{Chebyshev and minimax polynomials}
\label{sect:taylorvs}

Taylor polynomials minimize error near the expansion point,
but Chebyshev or minimax approximations
give better accuracy uniformly on an interval.

\begin{table}
\setlength{\tabcolsep}{3pt}
\renewcommand{\arraystretch}{1.1}
\centering
\caption{Maximum absolute error (estimated) of degree-$D$ polynomial approximations of $1 / \Gamma(x)$ on short real intervals.}
\label{table:tabpolyapprox}
\small
\begin{tabular}{l | l l l l}
 $D$ & Taylor, $[0.5, 1.5]$ & Taylor, $[0.75, 1.25]$ & Chebyshev, $[1,2]$ & Minimax, $[1,2]$ \\
 \hline
 $5$  & $2.1 \cdot 10^{-4}$  &  $2.8 \cdot 10^{-6}$  &  $3.4 \cdot 10^{-6}$  &  $3.8 \cdot 10^{-7}$ \\
 $10$  & $1.0 \cdot 10^{-8}$  &  $4.9 \cdot 10^{-12}$  &  $7.1 \cdot 10^{-12}$  &  $5.9 \cdot 10^{-13}$ \\
 $50$  & $1.2 \cdot 10^{-59}$  &  $5.0 \cdot 10^{-75}$  &  $4.3 \cdot 10^{-75}$  &   \\
 $100$  & $2.0 \cdot 10^{-139}$  &  $7.7 \cdot 10^{-170}$  &  $6.3 \cdot 10^{-170}$  &   \\
 $200$  & $4.2 \cdot 10^{-323}$  &  $1.3 \cdot 10^{-383}$  &  $9.2 \cdot 10^{-384}$  &   \\
 $500$  & $2.8 \cdot 10^{-964}$  &  $4.3 \cdot 10^{-1115}$  &  $2.1 \cdot 10^{-1117}$  &   \\ [1ex]
\end{tabular}
\end{table}

True minimax polynomials are optimal by definition,\footnote{For simplicity, we consider only polynomial approximants. Rational
functions of degree $P + Q = D$ can generally achieve higher accuracy
than polynomials of degree $D$,
but the improvement is small when working with
an entire function such as the reciprocal gamma function.}
but they are expensive to compute,
which largely limits their usefulness to machine precision.
\emph{Chebyshev interpolants}, however, are easy to generate
while being nearly as accurate as minimax polynomials
of the same degree~\cite[Section 3.11]{DLMF} \cite[Theorem 16.1]{trefethen2019approximation}.

For any $N$-times differentiable function $f$ on an interval $[a,b]$, let $f_N(x)$ be the degree-($N-1$)
Lagrange interpolating polynomial of $f$ at the Chebyshev nodes
\begin{equation}
x_k = \frac{1}{2}(a+b) + \frac{1}{2} (b-a) \cos\!\left(\frac{2k-1}{2N}\pi\right), \quad 1 \le k \le N.
\end{equation}
Then the standard error bound
\begin{equation}
|f(x) - f_{N}(x)| \le \frac{1}{2^{N-1} N!} \left( \frac{b-a}{2} \right)^N \max_{\xi \in [a,b]} | f^{(N)}(\xi) |, \quad x \in [a, b]
\end{equation}
suggests that, for an interval of width $b-a = 1$,
Chebyshev interpolants will have worst-case error
about $2^{-N}$ times that of Taylor polynomials of the same degree
(in the worst case of radius $1/2$ for the Taylor series).
The Taylor series will
break even with the Chebyshev interpolant for the average input (radius $1/4$).
Numerical values (Table~\ref{table:tabpolyapprox}) confirm this.

Chebyshev interpolants are thus clearly better if we want
to achieve a fixed accuracy on a real interval.
However, there are some points in favor of using the Taylor series:
\begin{itemize}
\item The Taylor series
can easily be truncated according to the precision as well as
the proximity to the expansion point,
and will then perform nearly as well
as the Chebyshev interpolant \emph{on average}.
In practice, input to the gamma function is not necessarily uniform:
arguments close to integers often appear in perturbation problems,
favoring the Taylor series.
\item Computing the truncated Taylor series to length $N$ turns out
to be somewhat cheaper than evaluating $1/\Gamma(x)$ at $N$ points.
In fact, the most efficient way to construct Chebyshev interpolants
for the reciprocal gamma function appears to be to compute
Taylor polynomials and evaluate these.
\item The Taylor series is better if we want to to extend the approximation on $[a,b]$ to complex arguments inside a box $[a,b] + [-c,c] i$. (On the other hand, if we want to optimize for such a box with large $c$, then a Chebyshev interpolant along the imaginary axis might be better.)
\end{itemize}


\section{Hypergeometric methods}

\label{sect:hypergeometric}

The methods discussed so far require $\bigOtilde(p^2)$ bit operations
to compute $\Gamma(z)$ to $p$-bit accuracy
assuming that $z$ is fixed.
The methods that will be discussed next are interesting since they lead to improved asymptotic complexity bounds.
The two main theorems on the complexity
of gamma function computation are due, in respective order, to Brent~\cite{Brent1976} and Borwein~\cite{Borwein1987,borwein1987pi}.

\begin{theorem}
Let $z \not \in \{0, -1, -2, \ldots\}$ be a fixed algebraic number.
Then $\Gamma(z)$ can be computed to $p$-bit accuracy using $O(M(p) \log^2p) = \bigOtilde(p)$ bit operations.
\label{thm:complexityalg}
\end{theorem}

\begin{theorem}
Let $z \not \in \{0, -1, -2, \ldots\}$ be a fixed complex number.
Then $\Gamma(z)$ can be computed to $p$-bit accuracy using $O(p^{1/2} \log^2 p)$ arithmetic operations,
or $\bigOtilde(p^{3/2})$ bit operations
(assuming that we can compute $z$ to $p$-bit accuracy at least this fast).
\label{thm:complexitygen}
\end{theorem}

There are several
alternative hypergeometric representations
of the gamma function, any of which can
be used to prove Theorems~\ref{thm:complexityalg} and \ref{thm:complexitygen}.
All known series of this type
converge at an ordinary geometric rate, requiring
a number of terms proportional to $p$,
and the subquadratic complexity
simply results from exploiting the structure of hypergeometric terms.
In fact, the constant-factor overheads
are generally worse than those for the best global and local methods
discussed earlier, and the subquadratic hypergeometric
methods therefore only become competitive at extremely high precision
(thousands of digits or more)
where asymptotically fast arithmetic truly kicks in.
One of our goals in this section will be to
analyze constant factors
to determine the most viable hypergeometric formulas.

We will generally assume that the argument
of the gamma function is held fixed
in order to analyze asymptotics with respect to the precision.
For large enough $p$,
the reduced-complexity hypergeometric methods should be faster
than the Stirling series in some region around the origin
(similar to Figure~\ref{fig:taylordiamond})
whose boundaries grow
with $p$, but we will
not analyze the asymptotics of this region.

\subsection{General techniques and error bounds}

We recall that any hypergeometric series can be expressed in terms of
the generalized hypergeometric function
\begin{equation}
{}_pF_q(a_1,\ldots,a_p;b_1,\ldots,b_q;z) =
\sum_{n=0}^\infty \frac{(a_1)_n\dots(a_p)_n}{(b_1)_n\dots(b_q)_n} \frac {z^n} {n!}.
\label{eq:hyppfq}
\end{equation}

If $p < q + 1$ or if $p = q + 1$ and $|z| < 1$,
the series converges at the same asymptotic rate
as the sum of $z^n / (n!)^{q + 1 - p}$.
Otherwise, it diverges
at the rate of $z^n (n!)^{p - q - 1}$
and must be interpreted in the sense of analytic continuation
or asymptotic resummation.
We will define the ${}_2F_0$ function
via Kummer's $U$-function as ${}_2F_0(a,b,z) = (-1/z)^a U(a, a-b+1, -1/z)$.

To keep things brief, we will analyze the asymptotic efficiency
of most formulas below without giving
explicit error bounds.
Tail bounds for convergent series can be obtained easily
by comparing with geometric series;
see for example \cite[Theorem 1]{Johansson2019hypergeometric},
which also deals with parameter derivatives.
Tail bounds for ${}_2F_0(a,b,z)$
valid for complex $a,b,z$ are given by
Olver~\cite[section 13.7]{DLMF} \cite{Olver1965};
in the special case of positive real $a, b$ and negative real $z$,
the error is bounded by first omitted term.

The fast methods for evaluating hypergeometric series are
generalizations of the methods for rising factorials discussed in section~\ref{sect:rising}.
Denoting the $n$th term in \eqref{eq:hyppfq} by $t(n)$ and letting $r(n) = t(n) / t(n-1)$
denote the term ratio (a rational function of $n$), the $n$th partial sum
can be written as a matrix product
\begin{equation}
\begin{pmatrix} r(n) & 0 \\ 1 & 1 \end{pmatrix}
\begin{pmatrix} r(n-1) & 0 \\ 1 & 1 \end{pmatrix}
\cdots
\begin{pmatrix} r(1) & 0 \\ 1 & 1 \end{pmatrix}
=
\begin{pmatrix} \prod_{k=0}^{n} r(k) & 0 \\ \sum_{k=0}^{n-1} \prod_{j=0}^k r(j) & 1 \end{pmatrix}.
\end{equation}

This product can be computed using binary splitting (used in Theorem~\ref{thm:complexityalg})
or fast multipoint evaluation (used in Theorem~\ref{thm:complexitygen});
rectangular splitting is also viable~\cite{Johansson2014rectangular}, but does not
improve the asymptotic bit complexity.

Hypergeometric sequences and functions should
be understood as a special case of holonomic sequences and functions.
A sequence is called holonomic if it
satisfies a linear recurrence relation of order $r$ with rational
function coefficients (a hypergeometric sequence has order $r = 1$);
the sum of an order-$r$ sequence can be evaluated
as a product of size $(r+1) \times (r+1)$ matrices,
using the same methods.

\subsubsection{Derivatives}

Any hypergeometric representation of the gamma function
can be differentiated to yield representations of the digamma function $\psi(z)$
and higher derivatives.
Indeed, for any fixed $n \ge 0$, Theorems~\ref{thm:complexityalg} and \ref{thm:complexitygen}
are valid with $\Gamma^{(n)}(z)$ or 
$\psi^{(n)}(z)$ instead of $\Gamma(z)$.\footnote{At least in the sense of achieving $p$-bit absolute accuracy;
if $z$ is given only as a black box approximation program, we must exclude points where $\psi^{(n)}(z) = 0$
in order to guarantee convergence to a relative accuracy.}
The differentiated series are holonomic but not generally hypergeometric.

Explicit formulas for parameter derivatives of hypergeometric
functions are typically not pretty,
but as discussed in~\cite{Johansson2019hypergeometric},
an implementation can simply use the original formulas
in combination with
power series arithmetic.
When power series arithmetic is combined with binary splitting
this also leads to asymptotically fast algorithms
for higher derivatives;
in particular, the complexity is
$\bigOtilde(n^2)$ for computing the first $n$ derivatives
simultaneously to precision $p = \bigOtilde(n)$.

Some of the following formulas have removable singularities at
certain points; the appropriate limits can be computed
in the same manner as derivatives.

\subsection{Using incomplete gamma functions}


One well-known method to compute the gamma function \cite{Brent1976,Brent1978,borwein1987pi}
uses the decomposition
\begin{equation}
\Gamma(z) = \gamma(z,N) + \Gamma(z,N)
\label{eq:gammadecomp}
\end{equation}
where
\begin{equation}
\gamma(z,N) = \int_0^N t^{z-1} e^{-t} dt, \quad \quad \Gamma(z,N) = \int_N^{\infty} t^{z-1} e^{-t} dt
\end{equation}
are the lower and upper incomplete gamma functions (assuming $N > 0$, $\operatorname{Re}(z) > 0$).
The free parameter $N$ may be taken to be an integer
for efficient evaluation.

Since $\Gamma(z,N) \approx e^{-N}$, we have 
$\Gamma(z) \approx \gamma(z,N)$ to within $2^{-p}$
if $N \approx \log(2) p$.

The lower incomplete gamma function can be computed using
\begin{equation}
\gamma(z,N) \,=\, \frac{N^z e^{-N}}{z} {}_1F_1(1, z+1, N) \,=\, \frac{N^z e^{-N}}{z} \sum_{k=0}^{\infty} \frac{N^k}{(z+1)_k}.
\label{eq:lowerseries}
\end{equation}
If $N \approx \log(2) p$, then an asymptotic analysis
using $(z+1)_k \approx k! \approx (k/e)^k$
shows that the total error in the approximation of $\Gamma(z)$ is of order $2^{-p}$
when the series is truncated after $e \log(2) p \approx 1.88417 p$ terms.

The formula \eqref{eq:gammadecomp} becomes slightly more efficient
if we compute
$\Gamma(z,N)$ instead of neglecting this term, using
the asymptotic expansion
\begin{equation}
\Gamma(z,N) \, = \, e^{-N} N^{z-1} {}_2F_0\left(1, 1-z, -\tfrac{1}{N}\right) \,\sim \, e^{-N} N^{z-1} \sum_{k=0}^{\infty} \frac{(1-z)_k}{(-N)^k}.
\label{eq:upperseries}
\end{equation}

Let $N = \alpha \log(2) \, p$ where $\alpha$ is a tuning parameter,
which will be used to control the number of terms in $\eqref{eq:lowerseries}$
and $\eqref{eq:upperseries}$.
Since $\eqref{eq:upperseries}$ is a divergent series,
the accuracy that can be achieved is limited for any fixed $N$;
we will need $\alpha \ge 0.5$ to ensure that we can reach
an error of $2^{-p}$.

Suppose that $\eqref{eq:lowerseries}$ is truncated after $K$ terms
and that $\eqref{eq:upperseries}$ is truncated after $L$ terms,
and that the truncation errors are of the same order of magnitude
as the $K$-th and $L$-th terms of the respective series.
Then the errors $\varepsilon_{\gamma}$ and $\varepsilon_{\Gamma}$
for the respective terms in \eqref{eq:gammadecomp} are of order
\begin{equation}
\varepsilon_{\gamma} \approx e^{-N} \frac{N^K}{K!}, \quad \varepsilon_{\Gamma} \approx e^{-N} \frac{L!}{N^L}.
\end{equation}
We can solve $\varepsilon_{\gamma} = \varepsilon_{\Gamma} = 2^{-p}$ asymptotically
for $K$ and $L$ by substituting $k! \approx (k/e)^k$,
taking logarithms, and employing
the Lambert $W$ function.
We obtain
\begin{equation}
K = \left(\frac{1-\alpha}{W_0\!\left(\frac{1-\alpha}{\alpha e}\right)} \right) \log(2) p,
\quad L = \left( \frac{\alpha-1}{W_{-1}\!\left(\frac{\alpha-1}{\alpha e}\right)} \right) \log(2) p
\end{equation}
where the $W_{-1}$ branch of the Lambert $W$ function is used for $L$ since we want
the solution where the terms of the asymptotic series are decreasing
and not increasing.
The function $W_{-1}(x)$ is real for $-e^{-1} \le x < 0$;
the lower inequality reflects the constraint $\alpha \ge 0.5$
required to achieve sufficient accuracy with the asymptotic series.

In the limit $\alpha \to 1$, we obtain $K = e \log(2) p \approx 1.88417 p$ and $L = 0$, minimizing
the number of terms of the asymptotic series for $\Gamma(z,N)$.

In the limit $\alpha \to 0.5$, we obtain $K = p \log(2) / (2 W(1/e)) \approx 1.24459 p$
and $L = p \log(2) / 2 \approx 0.346574 p$.
This minimizes the number of terms of the series of $\gamma(z,N)$.
The total number of terms for both series is $K + L \approx 1.59116 p$.

We can minimize the total number of terms $K + L$ numerically;
the minimum occurs at $\alpha = \alpha_0 \approx 0.546904$,
where $K \approx 1.30970p$, $L \approx 0.179996p$ and $K + L \approx 1.48970p$.
In practice, it is probably best to take a smaller $0.5 < \alpha < \alpha_{0}$
since the series for $\gamma(z,N)$ must be computed with higher precision
than that for $\Gamma(z,N)$ and thus requires more work for each term.

\subsubsection{More points}

The method using incomplete gamma functions amounts
to computing the defining
integral~\eqref{eq:gammadef} using term-by-term integration of
series expansions at $0$ and $\infty$.
An idea is to expand the integrand at multiple points
$0 < t_1 < \ldots < t_n < \infty$
using
\begin{equation}
(t+x)^{z-1} e^{-(t+x)} = t^{z-1} e^{-t} \sum_{n=0}^{\infty} {}_1F_1(-n, z-n, t) {z-1 \choose n} t^{-n} x^n
\end{equation}
in which the series is holonomic.
Unfortunately, this does not seem to lead to an improvement;
the reduction in terms at $0$ and $\infty$ will not
be significant unless we use vastly more terms in the middle.

\subsection{Using Bessel functions}

An alternative method uses the
modified Bessel function $I_{\nu}$.\footnote{The ordinary Bessel function $J_{\nu}$ will also do the job, but less efficiently since
the corresponding convergent expansion suffers from catastrophic cancellation.}
We will need the convergent expansion
\begin{equation}
\Gamma(\nu+1) I_{\nu}(x) = \left(\frac{x}{2}\right)^{\nu} {}_0F_1\left(\nu+1, \frac{x^2}{4}\right)
=
\left(\frac{x}{2}\right)^{\nu} \sum_{k=0}^{\infty} \frac{1}{k! (\nu+1)_k} \left(\frac{x^2}{4}\right)^{k},
\label{eq:besseliconvergent}
\end{equation}
and the asymptotic expansion
\begin{equation}
I_{\nu}(x) = \frac{1}{\sqrt{2 \pi x}} \left(
    e^x {}_2F_0\left(\tfrac{1}{2}+\nu, \tfrac{1}{2}-\nu, \tfrac{1}{2x}\right) -
    i e^{-x-\pi i \nu} {}_2F_0\left(\tfrac{1}{2}+\nu, \tfrac{1}{2}-\nu, -\tfrac{1}{2x}\right) \right)
\label{eq:besseliasymptotic}
\end{equation}
valid at least for $x > 0$.
We obtain $\Gamma(z+1)$ if we
choose a large enough $N$, compute
$\Gamma(z+1) I_{z}(N)$ using \eqref{eq:besseliconvergent}
and $I_{z}(N)$ using \eqref{eq:besseliasymptotic},
and divide.

Suppose that $\eqref{eq:besseliconvergent}$ is truncated after $K$ terms
and that the series with prefactor $e^{x}$ in $\eqref{eq:besseliasymptotic}$ is truncated after $L$ terms,
and that the truncation errors are of the same order of magnitude
as the $K$-th and $L$-th terms of the respective series (we can ignore the ${}_2F_0$ series
in $\eqref{eq:besseliasymptotic}$ with exponentially small prefactor).
Then the \emph{relative} errors $\varepsilon_{0}$ and $\varepsilon_{\infty}$
are of order
\begin{equation}
\varepsilon_{0} \approx \frac{e^{-N}}{(K!)^2} \left(\frac{N^2}{4}\right)^K, \quad \varepsilon_{\infty} \approx \frac{L!}{(2N)^L}.
\end{equation}
We assume that $N \approx \alpha \log(2) p$ for some tuning parameter $\alpha$.
Solving asymptotically gives
\begin{equation}
K = \left(\frac{1-\alpha}{2 W_0\!\left(\frac{1-\alpha}{\alpha e}\right)}\right) \log(2) p, \quad
L = \left(-\frac{1}{W_{-1}\!\left(-\frac{1}{2 \alpha e}\right)}\right) \log(2) p.
\end{equation}
Real-valuedness of $W_{-1}(x)$
reflects the constraint $\alpha \ge 0.5$
required to achieve sufficient accuracy with the asymptotic series.

We can minimize $K$ by setting $\alpha = 0.5$, giving
$K \approx p \log(2) / (4 W(1/e)) \approx 0.622294 p$
and $L \approx \log(2) p \approx 0.693147 p$,
with $K + L \approx 1.31544 p$.

Minimizing $K + L$ gives $\alpha = \alpha_0 \approx 0.639845$,
with $K \approx 0.717353 p$, $L \approx 0.369583 p$ and $K + L \approx 1.08694 p$.
This is not necessarily the optimal strategy since the
two series have somewhat different terms; a more
realistic analysis should account for the relative costs
of computing the terms of each series.

\subsubsection{Alternative Bessel function algorithm}
We can exploit the connection formula
\begin{equation}
I_{-\nu} = I_{\nu} +  \frac{2 \sin(\pi \nu)}{\pi} K_{\nu}(x).
\end{equation}
involving the modified Bessel function of the second kind
\begin{equation}
K_{\nu}(x) = \left(\frac{2 x}{\pi}\right)^{-1/2} 
    e^{-x} {}_2F_0\left(\tfrac{1}{2}+\nu, \tfrac{1}{2}-\nu, -\tfrac{1}{2x}\right).
\end{equation}
Let $A = \Gamma(z+1) I_z(N)$ and $B = \Gamma(-z+1) I_{-z}(N)$,
both of which can be computed using \eqref{eq:besseliconvergent}.
An application of the connection formula and
the reflection formula for the gamma function gives
\begin{equation}
\label{eq:gammabessel2}
\Gamma(z)^2 = \frac{A}{B} \frac{\pi}{z \sin(\pi z)} \left(1 + \frac{2}{\pi \sin(\pi z)} \frac{K_z(N)}{I_z(N)}\right).
\end{equation}
To compute $\Gamma(z)$, we can determine the correct
sign of the square root from a low-precision
approximation of $\Gamma(z)$ obtained by a different method
(when $z > 0$, we simply take the positive root).

We can choose $N \approx \tfrac{1}{2} \log(2) p$ so that $K_z(N) / I_z(N) \approx e^{-2N}$ is negligible,
resulting in the approximation
\begin{equation}
\Gamma(z) \approx \pm \sqrt{\frac{A}{B} \frac{\pi}{z \sin(\pi z)} }.
\label{eq:potter}
\end{equation}
In this case, we need $K \approx \tfrac{1}{2} e \log(2) p$ terms
for each of the series $A$ and $B$, or $2K \approx e \log(2) p \approx 1.88417 p$ terms in total.
Bailey \cite{bailey2021} credits
the ``very efficient but little-known formula''~\eqref{eq:potter}
to Potter~\cite{potter2014} and reports using it
in MPFUN package.\footnote{Bailey erroneously writes \eqref{eq:potter} with an equals sign.}

Alternatively, we can choose a somewhat smaller $N$ and compute
$K_z(N) / I_z(N)$ using asymptotic expansions.
Setting $N \approx \alpha \log(2) p$ and
solving asymptotically for
$e^{-2N} \frac{L!}{(2N)^L} = 2^{-p}$
gives
\begin{equation}
K = \left(\frac{1-\alpha}{2 W_0\!\left(\frac{1-\alpha}{\alpha e}\right)}\right) \log(2) p, \quad
L = \left(\frac{2\alpha-1}{W_{-1}\!\left(\frac{2 \alpha-1}{2 \alpha e}\right)}\right) \log(2) p.
\end{equation}
Because of the asymptotic series, we require $\alpha \ge 0.25$.

Minimizing $K$ gives $\alpha = 0.25$ with
$K \approx 3 \log(2) p / (8 W(3/e)) \approx 0.430672 p$
and $L \approx \tfrac{1}{2} \log(2) p \approx 0.346574 p$,
with $2K + 2L \approx 1.55449$ terms in total.

Minimizing $K + L$ gives $\alpha \approx 0.358032$,
where we need $2K + 2L \approx 1.16604 p$ terms in total.
As noted previously, this is not necessarily the
optimum in practice since the series have different terms
and the asymptotic series require lower precision.

\subsubsection{A third Bessel-type algorithm}

Another interesting approximate formula
is given by Smith~\cite[eq.\ 94]{smith2006gamma},
here slightly rewritten:
\begin{equation}
\Gamma(z)^2 \approx N^{2z} \left[- \frac{2 (\gamma + \log(N))}{z} \, {}_1F_2(z, 1, 1+z, N^2) + S_z(N) \right],
\end{equation}
\begin{equation}
S_z(N) = \sum_{n=0}^{\infty} \frac{(-N)^{2n}}{n!^2} \left( \frac{1}{(z+n)^2} + \frac{2 H_n}{n+z} \right), \quad H_n = \sum_{k=1}^n \frac{1}{k}.
\end{equation}
The error is around $e^{-2N}$, though Smith does not provide a derivation
of the formula or an error bound (we expect that this can be obtained
by differentiating some Bessel function identity with respect to parameters).
The function $S_z(N)$ is not hypergeometric, but it is holonomic.
Smith's formula appears to be less efficient
than the other formulas considered above, but it hints that there may be a space
of hypergeometric or holonomic approximation formulas for the gamma function
that has not yet been fully explored.


\subsection{Incomplete gamma versus Bessel}

Which hypergeometric method is superior: the one using incomplete
gamma functions or either algorithm using Bessel functions?
The first Bessel function method requires the fewest
terms, but the terms of the incomplete gamma functions
are simpler. The winner may ultimately depend on implementation details.

For evaluating $\Gamma(z)$ with $z \in \RR$
using the fast multipoint evaluation strategy ($\bigOtilde(p^{3/2})$ complexity),
the incomplete gamma function method runs around 30\% faster than
the Bessel function methods when implemented in Arb.
For an empirical comparison between hypergeometric series
and the Stirling series, see Section~\ref{sect:implresults}.

\subsection{Products and quotients}

Certain products and ratios of gamma functions
can be evaluated without computing
each gamma function separately.
Smith \cite{smith2006gamma} gives a collection
of formulas based on closed-form evaluations
of the Gauss hypergeometric function
${}_2F_1(a,b; c; x)$ at fixed points $|x| < 1$,
where we need $-1/\log_2(|x|)$ terms per bit of precision.
Notable examples include the beta function
\begin{equation}
B(a,b) = \frac{\Gamma(a) \Gamma(b)}{\Gamma(a+b)} = \frac{2}{2a+2b-1} \, \frac{{}_2F_1(1, 2a+2b-1;\, a+b+\tfrac{1}{2};\, \tfrac{1}{2})}{{}_2F_1(2a-1, 2b-1;\, a+b-\tfrac{1}{2};\, \tfrac{1}{2})},
\label{eq:smithbeta}
\end{equation}
the generalized reflection formula
\begin{equation}
\Gamma(M-z) \Gamma(M+z) = \frac{\sqrt{\pi} \, \Gamma(M+\tfrac{1}{2})}{{}_2F_1(2M-1-2z, 2M-1+2z;\, 2M-\tfrac{1}{2};\, \tfrac{1}{2})},
\label{eq:smithreflection}
\end{equation}
and the generalized shift (typo corrected)
\begin{equation}
\frac{\Gamma(z+S)}{\Gamma(z)} = \frac{1}{\sqrt{\pi}} \, \Gamma(S+\tfrac{1}{2}) \, {}_2F_1(2z-1, 2S;\, z+S;\, \tfrac{1}{2}).
\label{eq:smithshift}
\end{equation}

Each of the above ${}_2F_1$ series requires 1 term per bit of precision.
In the last two formulas, the gamma function values appearing as prefactors in the right-hand sides can
be precomputed for fixed $M$ or $S$; elliptic integral
identities (see below) can be used when $M+\tfrac{1}{2}$ or $S+\tfrac{1}{2}$
is a fraction $k/24$.
Even more rapidly convergent series exist for the special shift $1/6$~\cite{smith2006gamma,ekhad2004forty}:
\begin{equation}
\frac{\Gamma(z+\tfrac{1}{6})}{\Gamma(z)} = \frac{3 \, \Gamma(\tfrac{2}{3})}{2 \sqrt{3 \pi} \sqrt[3]{2}} \left(\frac{16}{27}\right)^{z}
{}_2F_1(\tfrac{1}{3} - 2z, \tfrac{2}{3} - 2z;\, 1-z;\, -\tfrac{1}{8}),
\end{equation}
\begin{equation}
\frac{\Gamma(z+\tfrac{1}{6})}{\Gamma(z)} = \frac{\Gamma(\tfrac{3}{4})}{\Gamma(\tfrac{7}{12}) \sqrt{2} \sin(\tfrac{5 \pi}{12})} (3-24 \eta)^{1/4-z} 
{}_2F_1(\tfrac{1}{2} - 2z, \tfrac{1}{2} - 2z;\, 1-z;\, \eta),
\end{equation}
where $\eta = (2 - \sqrt{3})/4 \approx 0.0669873$.

Such formulas are potentially useful for evaluating the gamma
function on a grid of equidistant points.

Smith proposes an algorithm
for reducing the argument of $\Gamma(1+z)$
to the strip $0 \le \operatorname{Re}(z) \le 1/48$
using generalized reflections or shifts (\ref{eq:smithreflection}, \ref{eq:smithshift})
and evaluation of elliptic integrals,
or more generally to an arbitrarily small box
$0 \le \operatorname{Re}(z) \le 1/N$,
$0 \le \operatorname{Im}(z) \le 1/N$
with the help of a table of precomputed gamma function values.
This can be used to improve convergence of methods such as Taylor series,
but it is not clear whether it is worth the overhead over more direct algorithms.

\subsection{Elliptic integrals and special values}

The constants $\Gamma(k/24)$ can be expressed in terms
of the complete elliptic integral
\begin{equation}
K(m) = \tfrac{1}{2} \pi \, {}_2F_1(\tfrac{1}{2}, \tfrac{1}{2}; \, 1; \, m)
\end{equation}
evaluated at fixed
algebraic points~\cite{chowla1967,BorweinZucker1992}.
For example,
\begin{align}
\Gamma(\tfrac{1}{2}) &= [2 K(0)]^{1/2} = \sqrt{\pi}, \\
\Gamma(\tfrac{1}{3}) &= 2^{7/9} 3^{-1/12} \pi^{1/3} [K(\tfrac{1}{2} - \tfrac{1}{4} \sqrt{3}))]^{1/3}, \\
\Gamma(\tfrac{1}{4}) &= 2 \pi^{1/4} [K(\tfrac{1}{2})]^{1/2}.
\end{align}

The $K(m)$ values can be evaluated using binary splitting.
Alternatively, $K(m)$ can be expressed in terms
of the arithmetic-geometric mean
\begin{equation}
\operatorname{agm}(a_0,b_0) = \lim_{n \to \infty} a_n = \lim_{n \to \infty} b_n, \quad \left(a_{n+1}, b_{n+1}\right) = \left(\tfrac{a_n+b_n}{2}, \sqrt{a_n b_n}\right)
\end{equation}
which
converges to $p$-bit accuracy after only $O(\log p)$ iterations.

We have, for example,
\begin{equation}
\Gamma(\tfrac{1}{4}) = \sqrt{\frac{(2\pi)^{3/2}}{\operatorname{agm}(1, \sqrt 2)}},
\end{equation}
and for $\pi$, the Brent-Salamin algorithm\footnote{Independently discovered by Brent~\cite{brent1976multiple} and Salamin~\cite{Salamin1976} in 1975;
also known as the Gauss–Legendre algorithm, as if the names
of Gauss and Legendre were not already overloaded.}
\begin{equation}
\pi = \Gamma(\tfrac{1}{2})^2 = \frac{4 (\operatorname{agm}(1, 1/\sqrt 2))^2}{1 - \sum_{n=0}^{\infty} 2^n (a_{n} - b_{n})^2}.
\end{equation}

%

The arithmetic-geometric mean iteration has
better arithmetic complexity than series evaluation,
and better bit complexity by a factor $O(\log p)$.
It is not necessarily faster than
binary splitting in practice if the hypergeometric series
have favorable constant factors.
All recent record computations~\cite{Yee2021} of $\pi = \Gamma(1/2)^2$ have used
the Chudnovsky series~\cite{ChudnovskyChudnovsky1988}
\begin{equation}
\frac{1}{\pi} = 12 \sum_{n=0}^{\infty} \frac{(-1)^n (6n)! (13591409 + 545140134n)}{(3n)!(n!)^3 640320^{3n + 3/2}},
\end{equation}
less elegantly written as a linear combination of two ${}_3F_2$ functions,
which requires only $1 / \log_2(640320^3 / 1728) \approx 0.0212$ terms
per bit, i.e.\ adding more than 14 decimal digits per term.

In a similar vein, Brown \cite{brown2009algorithm} has derived
several rapidly convergent hypergeometric series for $\Gamma(1/3)$ and $\Gamma(1/4)$,
including the elegant pair
\begin{align}
[\Gamma(\tfrac{1}{4})]^4 &= \frac{32 \pi^3}{\sqrt{33}} \, {}_3F_2 \left(\tfrac{1}{2}, \tfrac{1}{6}, \tfrac{5}{6}; \, 1, 1; \, \tfrac{8}{1331}\right), \\
[\Gamma(\tfrac{1}{3})]^6 &= \frac{12 \pi^4}{\sqrt{10}} \, {}_3F_2 \left(\tfrac{1}{2}, \tfrac{1}{6}, \tfrac{5}{6}; \, 1, 1; \, -\tfrac{9}{64000}\right), \label{eq:brown13}
\end{align}
requiring only $0.136$ and $0.0782$ terms per bit, respectively.
The formula \eqref{eq:brown13} is used in Arb to compute $\Gamma(1/3)$
as it was found to be faster than the
arithmetic-geometric mean.

Products of $\Gamma(x)$ evaluated at fixed rational $x$ can be algebraic numbers.
Using such identities,
Vidunas~\cite{vidunas2005expressions} shows that
any $\Gamma(k/n)$ where $n$ divides 24 or 60 can be expressed
in terms of algebraic numbers, rational powers of $\pi$,
and the set of numbers $\Gamma(x)$ with $x \in \{1/3, 1/4, 1/5, 2/5, 1/8, 1/15, 1/20, 1/24, 1/60, 7/60\}$.

An alternative approach to numerical calculation
of special gamma function
values is to use the $q$-series of elliptic modular forms.
For example,
\begin{equation}
\Gamma(\tfrac{1}{4}) = \sqrt{2} \pi^{3/4} \theta(i) = \sqrt{2} \pi^{3/4} \vartheta(e^{-\pi})
\label{eq:g14theta}
\end{equation}
in terms of the Jacobi theta function
\begin{equation}
\theta(\tau) = \vartheta(q) = 1 + 2 \sum_{n=1}^{\infty} q^{n^2}, \quad q = e^{\pi i \tau}.
\end{equation}

Using a simple recurrence relation for the powers of $q$,
the theta series can be evaluated using $O(p^{1/2})$ arithmetic
operations, which is better than the arithmetic complexity for hypergeometric series (though not as
good as the arithmetic-geometric mean).
This complexity can in fact be improved even further~\cite{enge2018short,Nogneng2017}.
The bit complexity is unfortunately not quasilinear in $p$ (binary splitting is not useful here),
but the implied constant factors are small.
Using \eqref{eq:g14theta}, in which $q \approx 0.0432$,
we need only $0.470 p^{1/2}$ terms for $p$-bit accuracy.
We can accelerate the convergence
further using modular transformations and special values
of the elliptic lambda function~\cite{Yi2004};
for example,
\begin{equation}
\Gamma(\tfrac{1}{4}) = \pi^{3/4} \left[\frac{6 \sqrt{5} \sqrt{1 + \sqrt{5}}}{3 + \sqrt{5} + \left(\sqrt{3} + \sqrt{5} + {60}^{1 / 4}\right) {\left(2 + \sqrt{3}\right)}^{1 / 3}}\right] \vartheta(e^{-45 \pi})
\label{eq:theta45}
\end{equation}
where $e^{-45 \pi} \approx 4 \cdot 10^{-62}$,
requiring $0.0701 p^{1/2}$ terms.
The first 128 terms of this theta series suffice to determine $\Gamma(1/4)$ to one million decimal digits.\footnote{We are cheating by ignoring
the cost of computing $\pi$ and $e^{\pi}$, which are essentially as hard to compute as $\Gamma(1/4)$. The numbers $\Gamma(1/4)$, $\pi$ and $e^{\pi}$ are algebraically independent by a famous result of Nesterenko~\cite{Nesterenko1996}.
Formulas like~\eqref{eq:theta45} suggest that these numbers are in some sense ``nearly algebraically dependent'';
given any two of the numbers, we can construct extremely accurate approximations of the third using algebraic operations.}

\subsection{The digamma function}

The digamma function $\psi(z)$ can be computed
by evaluating a formula for $\Gamma(z)$ using length-two power series arithmetic (dual numbers),
but it is useful to have more direct formulas.

The best available method for Euler's constant $-\psi(1) = \gamma$
is the Brent-McMillan
algorithm~\cite{BrentMcMillan1980} based on Bessel functions. It uses the formula
\begin{equation}
\gamma = \frac{S_0(2N) - K_0(2N)}{I_0(2N)} - \log(N)
\label{eq:brentmcmillan}
\end{equation}
where the right-hand side should be evaluated using binary splitting
applied to the three series (the first of which is holonomic)
\begin{equation}
S_0(x) = \sum_{k=0}^{\infty} \frac{H_k}{(k!)^2} \left(\frac{x}{2}\right)^{2k}, \quad
I_0(x) = \sum_{k=0}^{\infty} \frac{1}{(k!)^2} \left(\frac{x}{2}\right)^{2k},
\end{equation}
\begin{equation}
I_0(x) K_0(x) \sim \frac{1}{2x} \sum_{k=0}^{\infty} \frac{[(2k)!]^3}{(k!)^4 8^{2k} x^{2k}}.
\end{equation}

A convenient error bound is available: if the first two series are summed up to the $k = K-1$ term
inclusive and the third series is summed up to $k = 2K-1$ inclusive,
where $K \ge \alpha N + 1$, $\alpha = 3 / W_0(3/e) \approx 4.97063$,
then the error in \eqref{eq:brentmcmillan} is bounded by $24 e^{-8N}$~\cite{BrentJohansson2013bound}.
We need about $\tfrac{1}{2} \alpha \log(2) p \approx 1.72269 p$ terms in total for $p$-bit accuracy.
This actually overestimates the true cost
because the first two series can share some of the
computations and the third series only needs to be computed with $p/2$-bit precision.

Used in combination with recurrence relations and Gauss's digamma theorem
\begin{equation}
\psi\!\left(\frac{k}{q}\right) = -\gamma - \log\!\left(2 q\right) - \frac{\pi}{2} \cot\!\left(\frac{\pi k}{q}\right) + 2 \sum_{n=1}^{\left\lfloor \left( q - 1 \right) / 2 \right\rfloor} \cos\!\left(\frac{2 \pi n k}{q}\right) \log\!\left(\sin\!\left(\frac{\pi n}{q}\right)\right)
\end{equation}
($q \ge 2$, $1 \le k \le q - 1$),
this also yields the most efficient method to
compute $\psi(k/q)$ for any small integers $k, q$.

The digamma function of a general\footnote{The singularities at positive integers and at $z = \pm \sqrt{5}$ are removable.} argument $z$ can be represented
by means of the constant $\gamma$, elementary functions, and a single rapidly convergent hypergeometric series:
\begin{equation}
\psi(z) = -\gamma - \frac{1}{2z} - \frac{\pi \cot(\pi z)}{2} - z^2 \left[ \frac{(z^2-5)}{4 (z^2-1)} \, {}_8F_7 ( a; b; -\tfrac{1}{4} ) \right]
\label{eq:zetagen1}
\end{equation}
where
\begin{center}
$a = (1, 1, 1-z, 1-z, 1+z, 1+z, 2+\tfrac{z}{\sqrt{5}}, 2-\tfrac{z}{\sqrt{5}})$,

$b = (\tfrac{3}{2}, 2, 2, 2+z, 2-z, 1+\tfrac{z}{\sqrt{5}}, 1-\tfrac{z}{\sqrt{5}})$.
\end{center}

The terms of the ${}_8F_7$ series satisfy the recurrence relation
\begin{equation}
\frac{t(k)}{t(k-1)} = \frac{k (k^2 - z^2)^2 (z^2 - 5 (k+1)^2)}{2(k+1)^2 (2k+1) ((k+1)^2-z^2) (5k^2-z^2)},
\end{equation}
involving only rational numbers and $z^2$.
At 0.5 terms per bit, the convergence is more rapid than for any
of the known hypergeometric series for the gamma function,
although the terms are more complicated.

The function in brackets in \eqref{eq:zetagen1} is
the ordinary generating function 
\begin{equation}
\frac{z^2-5}{4 (z^2-1)} {}_8F_7 ( a; b; -\tfrac{1}{4} ) = Z(z^2); \quad Z(z) = \sum_{n=0}^{\infty} \zeta(2n+3) z^n.
\label{eq:zetagen2}
\end{equation}
of the odd integer zeta values $\zeta(3), \zeta(5), \ldots$.
The hypergeometric series is given by
Borwein, Bradley and Crandall (\cite{BorweinBradleyCrandall2000}, unnumbered equation between (62) and (63))
as
\begin{align}
Z(z^2) = \sum_{k=1}^{\infty} \frac{(-1)^{k+1}}{k^3 {2k \choose k}} \left(\frac{1}{2} + \frac{2}{1-z^2/k^2}\right) \prod_{j=1}^{k-1} \left(1 - \frac{z^2}{j^2} \right)
\label{eq:zetagen3}
\end{align}
from which one can derive the closed ${}_8F_7$ form or the alternative representation
\begin{align}
Z(z^2) & = \frac{{}_4F_3(1, 1, 1-z, 1+z;\, \tfrac{3}{2}, 2, 2;\, -\tfrac{1}{4})}{4} \\
       & - \frac{{}_3F_2(1-z, 1-z, 1+z;\, \tfrac{3}{2}, 2-z;\, -\tfrac{1}{4})}{2 z (z-1)} \\
       & - \frac{{}_3F_2(1-z, 1+z, 1+z;\, \tfrac{3}{2}, 2+z;\, -\tfrac{1}{4})}{2 z (z+1)}.
\end{align}

Borwein, Bradley and Crandall use symbolic differentiation of~\eqref{eq:zetagen3}
to derive binary splitting schemes for isolated integer zeta values.
They seem to have overlooked (or neglected to mention)
the fact that the same series can be used for asymptotically fast
multi-evaluation of integer zeta values
by using binary splitting together with power series arithmetic
(this is more efficient than the incomplete gamma function method proposed in~{\cite{BorweinBradleyCrandall2000}};
see the discussion in \cite{Johansson2014thesis}),
as well as the fact that
this series can be used for reduced-complexity evaluation of the general function values $\psi(z)$ and $\psi^{(m)}(z)$.

We have curiously not found any other hypergeometric
representations of $\psi(z)$ in the literature
apart from the well-known
\begin{equation}
\psi(z) = -\gamma + \left(z - 1\right) \,{}_3F_2(1, 1, 2 - z, 2, 2, 1), \quad \Real(z) > 0,
\label{eq:psihyp}
\end{equation}
where the ${}_3F_2$ series converges too slowly to
be usable for direct summation.
A fast algorithm can be constructed from~\eqref{eq:psihyp} 
by connecting the solutions at $x = 0$ and $x = 1$ of the differential
equation defining ${}_3F_2(-;-;x)$,
using the standard analytic continuation
method~\cite{vdH:hol,Mezzarobba2011} for holonomic functions.
We do not attempt to construct an explicit algorithm here.


\section{Implementation results and algorithm selection}

\label{sect:implresults}

We implemented the following algorithms in Arb (version 2.21)
as part of a complete rewrite (some 10,000 new lines of C code)\footnote{See the documentation \url{https://arblib.org/arb_hypgeom.html} (functions for real variables) and \url{https://arblib.org/acb_hypgeom.html} (complex variables).}
of the library's functions for
$\Gamma(z)$, $1/\Gamma(z)$ and $\log \Gamma(z)$
for real, complex and rational $z$:

\begin{itemize}
\item The Stirling series (together with Algorithm~\ref{alg:stirlingsum}).
\item Hypergeometric series (using the incomplete gamma function method with the combined ${}_1F_1$ and ${}_2F_0$ series).
\begin{itemize}
\item Fast multipoint evaluation for generic arguments.
\item Binary splitting for rational arguments.
\end{itemize}
\item The Taylor series.
\item Factorials and elliptic integrals in special cases.
\end{itemize}

The algorithms are available as separate
functions; the user-facing functions (\texttt{arb\_gamma}, etc.)
try to choose the best algorithm automatically.

The previous version mainly used a less optimized
implementation of the Stirling series
together with factorials and elliptic integrals for special cases.
We also implemented other algorithms (Spouge's formula, Bessel function formulas) for
timing comparisons, but did not include them in the library
after experiments clearly showed worse performance.

\subsection{Varying precision}

Table~\ref{tab:maintimings} and Figure~\ref{fig:gammatime} show
the performance of different algorithms
when the argument is a fixed small real number,
for a varying precision~$p$ (we display the number of decimals $d$, equivalent to binary precision $p = \log_2(10) d$).

\begin{table}
\setlength{\tabcolsep}{3.1pt}
\renewcommand{\arraystretch}{1.02}
\centering
\caption{Time (in seconds) to compute $\Gamma(x)$, $x \approx 1.3$, to $d$ decimal digits using different algorithms:
Spouge's formula (``first'' including cost of generating coefficients, ``repeated'' with coefficients cached), the Stirling series,
Taylor series,
hypergeometric series (using fast multipoint evaluation). Lower table: using hypergeometric series with binary splitting, with $x = 13/10$ given as an exact rational number.}
\label{tab:maintimings}
\small
\begin{tabular}{l l l l l l l l}
 $d$ & Spouge & Spouge & Stirling & Stirling & Taylor & Taylor & Hyper- \\
     & (first)  & (repeated)  & (first)  & (repeated) & (first) & (repeated) & geometric  \\
 \hline
 $10$ $\phantom{2^{2^2}}$ & $2.2 \cdot 10^{-5}$ & $4.6 \cdot 10^{-6}$ & $5.2 \cdot 10^{-6}$  &  $4.1 \cdot 10^{-6}$   & $3.3 \cdot 10^{-5}$ & $2.6 \cdot 10^{-7}$ &  $4.7 \cdot 10^{-5}$ \\
 $30$ & $6.1 \cdot 10^{-5}$ & $1.2 \cdot 10^{-5}$ &   $1.3 \cdot 10^{-5}$  &  $8.8 \cdot 10^{-6}$                           &  $9.6 \cdot 10^{-5}$ & $1.8 \cdot 10^{-6}$  &  $0.00012$    \\
 $100$ & $0.00035$ & $5.2 \cdot 10^{-5}$ &  $5.2 \cdot 10^{-5}$  &  $2.9 \cdot 10^{-5}$                                 &  $0.00082$ & $8.6 \cdot 10^{-6}$       &  $0.00032$   \\
 $300$ & $0.0021$ & $0.00025$ &  $0.00038$  &  $9.6 \cdot 10^{-5}$                                                      &  $0.0077$ & $4.0 \cdot 10^{-5}$       &  $0.0011$    \\
 $1000$ & $0.029$ & $0.0027$ &  $0.0021$  &  $0.00073$                                                             &  $0.13$ & $0.00046$            &  $0.0064$    \\
 $3000$ & $0.47$ & $0.042$ &  $0.015$   &  $0.0064$               &  $1.7$    &  $0.0051$       &  $0.11$          \\
 $10000$ & $11$ & $1.0$ &  $0.20$   &  $0.087$                   &  $40$ &   $0.082$        &  $0.95$          \\
 $30000$ & $184$ & $16$ &  $2.3$    &  $1.0$                     &  $807$   &    $1.0$        &  $7.1$        \\
 $100000$ & $3683$ & $266$ &  $38$    &  $16$                      &      &            &  $67$        \\
 $300000$ &  &  &  $478$   &  $173$                      &    &            &  $481$        \\
 $1000000$ &  &  &         &                           &    &            &  $4108$         \\ [1ex]
\end{tabular}
\setlength{\tabcolsep}{8pt}
\begin{tabular}{l l l l l l}
 \multicolumn{6}{c}{Hypergeometric (rational $x = 13/10$, binary splitting)} \\
 $d$ &  & $d$ & & $d$ & \\
 \hline
 $10$  & $1.8 \cdot 10^{-5}$  & $1000$  & $0.0012$ &  $100000$ & $0.49$ \\
 $30$  & $3.4 \cdot 10^{-5}$  & $3000$  & $0.0045$ &  $300000$ &  $2.0$ \\
 $100$ & $9.7 \cdot 10^{-5}$  & $10000$ & $0.021$  &  $1000000$ & $7.9$ \\
 $300$ & $0.00030$            & $30000$ & $0.096$  &  $3000000$ & $31$
\end{tabular}
\end{table}

\begin{figure}
\caption{Time (in seconds) to compute $\Gamma(x)$, $x \approx 1.3$, to $d$ decimal digits using different algorithms (see Table~\ref{tab:maintimings}). \label{fig:gammatime}}
\includegraphics[width=10cm]{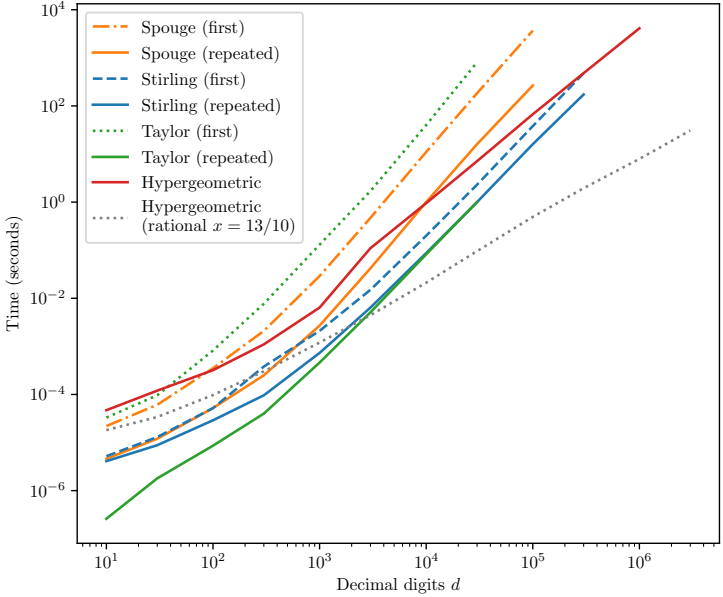}
\end{figure}

\subsubsection{Quadratic algorithms}

Spouge's formula, the Stirling series and the Taylor series
all have complexity $\bigOtilde(p^2)$
with and without precomputation, but with
different amounts of overhead.
These trends are clearly visible in~Figure~\ref{fig:gammatime}.

The Taylor series is the fastest method up to
perhaps $10^4$ decimal digits assuming that coefficients are precomputed.
The speedup over the Stirling series is roughly $15\times$
at 10 digits, $3.5\times$ at 100 digits, and $1.5\times$ at 1000 digits.
These figures are only rough estimates, and vary with the argument.
(The input represents a near-worst case
for the Stirling series, which performs
better with larger $|z|$; more on this below.)
In theory, the Taylor series should be $O(\log p)$
faster than the Stirling series asymptotically, but this is not
seen in practice.

The Taylor series has the slowest precomputation.
The precomputation time for 1000-digit precision would be roughly 0.1 seconds,
which is an unacceptably large overhead for a first call to the gamma function.
In the final implementation used in the library,
we decided not to compute Taylor coefficients at runtime.
Instead, we include a static 130~KB table of Taylor coefficients
($N = 536$, $p = 3456$), supporting precision up to around 1000 digits;
we simply avoid the Taylor series at higher precision
where the speedup over the Stirling series in any case would be small.
The size of this table is about half the total size of
tables used for elementary functions~\cite{Johansson2015elementary},
or 5\% of the library binary size.

Spouge's formula is, as expected, significantly slower than the Stirling
series. A more optimized implementation could perhaps
save a factor two, but this would not affect the conclusion.

The most interesting result is the low precomputation cost for
the Stirling series thanks to
the fast Bernoulli number generation (Algorithm~\ref{alg:bernoulli})
in combination with the improved Stirling series~(Algorithm \ref{alg:stirlingsum}).
Some authors have
argued that the Bernoulli numbers in the Stirling series are ``inconvenient''~\cite{smith2006gamma},
and in many earlier implementations, a first call to the gamma function
has indeed been 10 or 100 times slower than subsequent calls~\cite{laurie2005,Johansson2014rectangular}.
This is no longer the case: generating Bernoulli numbers now only costs
the equivalent of 1-2 extra evaluations of the gamma function.
What is notable is that we have achieved this without
sacrificing performance for repeated evaluations, which would be the tradeoff if we simply were to choose a very large (suboptimal) $\beta = r / p$.

\subsubsection*{Subquadratic algorithms}

Binary splitting evaluation of hypergeometric series
is highly efficient for rational arguments,
outperforming the Stirling and Taylor series from about 2000 digits.\footnote{We did not fully optimize the Stirling and
Taylor series for rational arguments, which should improve their efficiency. On the other hand, the hypergeometric series
evaluation could also be optimized by using a basecase algorithm specifically for low precision.}
Figure~\ref{fig:gammatime} clearly shows the $\bigOtilde(p)$ complexity.

The subquadratic ($\bigOtilde(p^{3/2})$) complexity for generic $z$
with the hypergeometric method
is also visible in Figure~\ref{fig:gammatime},
but it only breaks even with Stirling series at extremely
high precision: around $10^6$ bits if Bernoulli numbers have not been
cached, and at around $10^6$ digits ($p \approx 3 \cdot 10^6$) with Bernoulli numbers cached.

The high crossover with cached Bernoulli numbers reaffirms the
benchmark results in~\cite{Johansson2014rectangular},
but that study found a much lower crossover around $p$ = 30,000
for a first evaluation when Bernoulli numbers are not cached. This large
improvement for the initial call
to the Stirling series is due to the use of~Algorithm~\ref{alg:stirlingsum}.

Coincidentally, the crossover points are right around the range where
storing Bernoulli numbers in memory
starts to become an issue on a typical 2021-era laptop.
For example, with Algorithm~\ref{alg:stirlingsum},
we need 10 GB of Bernoulli numbers to use the
Stirling series for 600,000 digits of precision.
To use the Stirling series at higher precision,
it might be better to avoid caching Bernoulli
numbers and accept a factor 2-3 slowdown.

We did not optimize the hypergeometric series
evaluation at lower precision.
Using rectangular splitting instead of fast multipoint evaluation
should make the method more competitive with the quadratic
algorithms for $p < 10^4$~\cite{Johansson2014rectangular}.

\subsection{Varying argument}

When the variable $z$ as well as the precision $p$ are allowed
to vary, finding optimal cutoffs between several
algorithms is complicated.

\begin{figure}
\caption{Speedup using the Taylor series instead of the Stirling series to compute $\Gamma(x+yi)$, here at precision $p = 333$ (100 decimal digits). The
graphs end at the point where the default gamma function algorithm in Arb switches from the Taylor series to the Stirling series. The
axes of the figure are logarithmic. \label{fig:taylorspeed}}
\includegraphics[width=10cm]{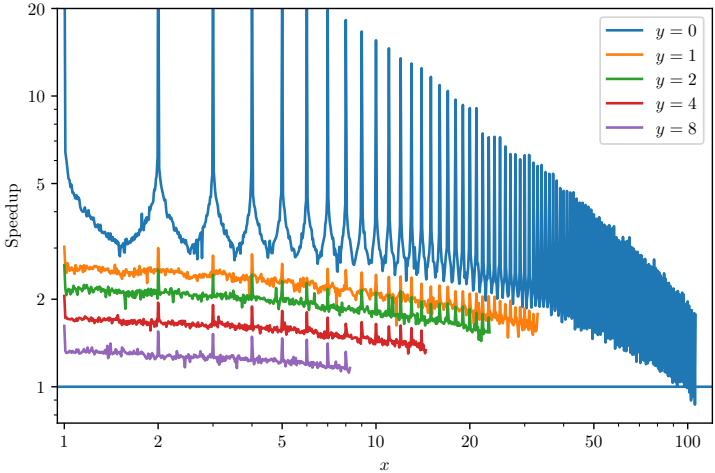}
\end{figure}

For switching between the Taylor series and the Stirling series,
we implemented experimentally-determined heuristics:

\begin{itemize}
\item If $z$ is real, let $r = \lfloor z + 1/2 \rfloor$.
      If $p < 40$ or if the radius of the ball for $z$ is greater than $2^{-16}$,
      we use the Taylor series if $|r| < 160$.
      Otherwise, we use the Taylor series if $-40 - (p - 40) / 4 < r < 70 + (p - 40) / 8$.
\item If $z$ is complex with $x = |\Real(z)|$, $y = |\Imag(z)|$, we do not use the Taylor series
      if either $p < 128$ and $y > 4$, $p < 256$ and $y > 5$,
      $p < 512$ and $y > 8$, $p < 1024$ and $y > 9$, or $y > 10$.
      If all these checks pass, we use the Taylor series if $x (1.0 + 0.75y) > 8 + 0.15p$.
\item In all cases, we still fall back on the Stirling series if $p$ or $|z|$
      are too large to achieve $p$-bit accuracy with the available precomputed Taylor coefficients
      (this is checked by computing the required working precision and number of terms
      and inspecting the table of coefficients).
\end{itemize}

Figure~\ref{fig:taylorspeed} demonstrates the
speedup over the Stirling series and the cutoffs for $p = 333$.
We can see that the Taylor series
performs better near integers
and gradually performs worse with increased $x$ and $y$.
The cutoffs are not quite optimal: we sometimes avoid the Taylor series
even where it would yield a $2\times$ speedup.
Overall, however, the tuning is reasonable, and we never use
the Taylor series where it would be more than a few percent
slower than the Stirling series.

We have also implemented experimentally-determined cutoffs for choosing
algorithms (Taylor series, the Stirling series,
elliptic integrals, hypergeometric series, factorials)
to compute $\Gamma(x)$ with rational $x = k/q$, where we must account
for the magnitude $|x|$ as well as the size
of the denominator $q$. We omit the details.\footnote{Interested readers may consult the source code
for the function \texttt{arb\_hypgeom\_gamma\_fmpq}.}

\subsection{Previous implementations}

Compared to the previous implementation of the gamma function in Arb,
the greatest speedup (typically a factor 2 to 10)
is observed where the Taylor series is used
for small real or complex variables.

Where the Taylor series is not used, our improvements to the Stirling series
often yields a factor 1.5 to 2 speedup
over the previous version
except for large arguments at low precision,
where the performance is unchanged.
For the first call (requiring Bernoulli numbers to be computed),
the speedup can be more than a factor 10.

With these improvements, the gamma function in Arb is
typically 5-6 times faster
than either Pari/GP 2.13~\cite{Par2021} (real and complex variables) or MPFR 4.1.0~\cite{Fousse2007} (real variables).
The speedup is greater at low precision (more than a factor 10 at machine precision)
and for a first call to the gamma function (at 10,000 digits, Arb is 1000 times
faster than MPFR).


\section{Open problems}

\label{sect:openprob}

We conclude with a collection of interesting problems
for future research.

\begin{problem}
What are tight bounds for the error terms in the Lanczos and Spouge formulas
(and other formulas of this type)? What are tight bounds for the error in the
Stieltjes continued fraction?
Can there be a unified theory for global approximations of the gamma function
that subsumes ad-hoc developments?
We note that Pugh~\cite[Chapter 11]{pugh2004analysis} gives some more specific
unsolved problems relating to the Lanczos approximation.
\end{problem}

\begin{problem}
Is it possible to re-expand the remainder term $R_N(z)$ in the Stirling series
in a way that reduces the computational complexity?
\end{problem}

\begin{problem}
Is $\bigOtilde(p^{3/2})$ the best possible complexity
for computing the gamma function to $p$-bit
accuracy at a generic point? Is there an algorithm with
quasilinear complexity?
\end{problem}

\begin{problem}
What is the complexity of computing $\Gamma(z)$ or $\Gamma^{*}(z)$ to $p$-bit accuracy accounting
for both $p$ and the value $z$? Is it possible to achieve subquadratic
complexity with respect to $p$ uniformly in $z$?
\end{problem}

\begin{problem}
In realistic computational complexity models,
what are the precise boundaries (as functions of both $p$ and $z$)
for optimal selection between the Taylor series,
the Stirling series, and other algorithms?
\end{problem}

\begin{problem}
In what ways can $\Gamma(z)$ and $\psi(z)$ be expressed
in terms of hypergeometric or holonomic functions with $z$ as a parameter?
For example, is it possible to express $\Gamma(z)$ itself
strictly in terms of hypergeometric functions ${}_{p}F_q(\{a_j\};\{b_j\};x)$
where the parameters $\{a_j\}$, $\{b_j\}$ are fixed rational functions of $z$
and where $x$ is a fixed point strictly within the radius of convergence of the ${}_{p}F_q$ series?
Are the known formulas for $\Gamma(z)$ and $\psi(z)$ the most efficient possible?
\end{problem}

\begin{problem}
Are there any rapidly convergent hypergeometric series
or arithmetic-geometric mean type formulas for $\Gamma(z)$
at specialized algebraic points $z$ other than rationals with denominators 24?
For example, can there be a formula similar to \eqref{eq:brown13} for $\Gamma(1/5)$ or $\Gamma(\sqrt{-3})$?
This would seemingly require entirely different mathematical
theory than the classical theory of elliptic integrals and elliptic modular forms.
\end{problem}

\begin{problem}
\label{problem:codegen}
Is it possible to generate efficient and formally verified
code for the gamma function automatically?
Generating useful code from first principles (e.g.\ with only the defining integral~\eqref{eq:gammadef} given as input)
seems to require symbolic knowledge about complex analysis and numerical techniques
far beyond the capabilities of current tools for automated theorem proving and code generation.
Current code generators for mathematical functions~\cite{kupriianova2014metalibm,lauter2015semi}
typically require a bounded domain and a function satisfying a nice differential equation,
and cannot be used for global computation of complex functions with essential singularities.
\end{problem}

\begin{problem}
What are the worst cases for Ziv's algorithm to ensure
correct rounding in different precision formats?

Bounding \emph{a priori} the precision required to ensure $p$-bit correct
rounding of transcendental functions is a hard theoretical problem.
Currently, the only solution is to perform an exhaustive
search for worst-case input for a fixed $p$,
which in turn reduces to a Diophantine approximation problem.
Brisebarre and Hanrot~\cite{brisebarre2021} discuss
state of the art methods for this problem
and consider the feasibility of finding worst cases for rounding $\Gamma(x)$ on the interval $[1, 2)$ with $p = 113$.
\end{problem}

\begin{problem}
To ensure that Ziv's strategy terminates, we need to be able
to check for representable (e.g.\ algebraic) points $z$ where $\Gamma(z)$ is exactly representable.
For example, in real binary floating-point arithmetic,
we need to detect $z  \in \ZZ[\tfrac{1}{2}]$ for which $\Gamma(z) \in \ZZ[\tfrac{1}{2}]$.
Conjecturally, the only points of this kind are trivial, i.e.\ the positive integers.
This conjecture can be generalized to the following questions:

\begin{itemize}
\item Are $\Gamma(z)$ and $\log \Gamma(z)$ transcendental for all $z \in \overline{\mathbb{Q}} \setminus \ZZ$?
\item Are $\psi(z)$, $\psi^{(n)}(z)$ and $\Gamma^{(n)}(z)$ $(n \ge 1)$ transcendental for all
$z \in \overline{\mathbb{Q}} \setminus \ZZ_{\le 0}$?
\end{itemize}

Currently, $\Gamma(z)$ is known to be transcendental for rational
arguments $k+\tfrac{1}{2}$, $k+\tfrac{1}{3}$, $k+\tfrac{1}{4}$ and $k+\tfrac{1}{6}$~\cite{waldschmidt2006transcendence}.
The special case of establishing the irrationality (and transcendence) of $-\psi(1) = \gamma$ is
of course famously open.
\end{problem}

\section{Acknowledgements}

The author thanks Daniel Schultz for complaining that
the gamma function in Arb was slow, which ultimately led us down this rabbit hole.

Karim Belabas pointed out a small improvement in Algorithm~\ref{alg:bernoulli} over
the original implementation in Arb.
Problem~\ref{problem:codegen} is inspired by a question posed by Jacques Carette.

The author was supported in part by the ANR grant ANR-20-CE48-0014-02 NuSCAP.

\bibliographystyle{amsalpha}
\bibliography{references}

\end{document}